\documentclass[11pt]{article}

\usepackage[numbers]{natbib}

\usepackage{amsmath}
\usepackage{amsthm}
\usepackage{amssymb}
\usepackage{algorithm}
\usepackage{algpseudocode}
\usepackage{subfig}

\usepackage{graphicx}
\usepackage{grffile}
\usepackage{wrapfig,epsfig}
\usepackage{url}
\usepackage[dvipsnames]{xcolor}
\usepackage{epstopdf}
\usepackage{tabularx}
\usepackage{enumerate}

\usepackage{bbm}
\usepackage{dsfont}

\usepackage{footmisc}

\usepackage{lipsum}
\usepackage[]{mdframed}

\allowdisplaybreaks

\let\C\relax
\usepackage{tikz}
\usetikzlibrary{quantikz2}
\usepackage{hyperref}
\hypersetup{colorlinks=true,citecolor=blue,linkcolor=red}

\usetikzlibrary{arrows}
\usepackage[margin=1in]{geometry}
\graphicspath{{./figs/}}

\graphicspath{{./figs/}}

\newtheorem{theorem}{Theorem}[section]
\newtheorem{lemma}[theorem]{Lemma}

\newtheorem{definition}[theorem]{Definition}
\newtheorem{notation}[theorem]{Notation}

\newtheorem{corollary}[theorem]{Corollary}

\newtheorem{fact}[theorem]{Fact}
\newtheorem{remark}[theorem]{Remark}

\newcommand{\wt}{\widetilde}

\newcommand{\eps}{\varepsilon}
\renewcommand{\epsilon}{\varepsilon}
\renewcommand{\phi}{\varphi}

\newcommand{\R}{\mathbb{R}}

\newcommand{\Z}{\mathbb{Z}}
\newcommand{\F}{\mathbb{F}}
\newcommand{\C}{\mathbb{C}}

\newcommand{\polylog}{\mathrm{polylog}}

\renewcommand{\tilde}{\wt}

\newcommand{\poly}{\mathrm{poly}}

\newcommand{\dist}{\mathrm{dist}}

\newcommand{\Id}{\mathsf{I}}

\DeclareMathOperator*{\E}{{\mathbb{E}}}

\newcommand{\QIPCP}{\textsc{QIPCP}}
\newcommand{\QMA}{\textsc{QMA}}
\newcommand{\NP}{\textsc{NP}}
\newcommand{\LH}{\textsc{LH}}

\newcommand{\rd}{r}

\newcommand{\sW}{\mathsf{W}}
\newcommand{\sV}{\mathsf{V}}
\newcommand{\sM}{\mathsf{M}}
\newcommand{\sS}{\mathsf{S}}

\newcommand{\tD}{\texttt{D}}
\newcommand{\tE}{\texttt{E}}
\newcommand{\sC}{\mathsf{C}}

\renewcommand{\Im}{\mathrm{Im}}

\makeatletter
\newcommand*{\RN}[1]{\expandafter\@slowromancap\romannumeral #1@}
\makeatother

\definecolor{b2}{RGB}{51,153,255}
\definecolor{mygreen}{RGB}{80,180,0}

\usepackage[normalem]{ulem}
\newif\ifshowrevisions
\showrevisionsfalse
\ifshowrevisions
\definecolor{revisionold}{RGB}{170,40,40}
\definecolor{revisionnote}{RGB}{35,100,180}
\definecolor{revisionfinal}{RGB}{40,130,60}
\newcommand{\oldfrag}[1]{\ifmmode{\color{revisionold}#1}\else{\color{revisionold}{\scriptsize\textbf{[orig]\,}}\uline{#1}}\fi}
\newcommand{\notefrag}[1]{\ifmmode\text{{\small\color{revisionnote}[note: #1]}}\else{\color{revisionnote}{\scriptsize\textbf{[note]\,}}\uline{#1}}\fi}
\newcommand{\finalfrag}[1]{\ifmmode{\color{revisionfinal}#1}\else{\color{revisionfinal}{\scriptsize\textbf{[final]\,}}\uline{#1}}\fi}

\newcommand{\origblock}[1]{\noindent{\scriptsize\textbf{\color{revisionold}[orig]\,}}{\color{revisionold}#1}}
\newcommand{\origitem}[1]{\item {\color{revisionold}{\scriptsize\textbf{[orig]\,}}#1}}
\newcommand{\noteblock}[1]{\noindent{\scriptsize\textbf{\color{revisionnote}[note]\,}}{\color{revisionnote}#1}}

\else
\definecolor{revisionold}{RGB}{0,0,0}
\definecolor{revisionnote}{RGB}{0,0,0}
\definecolor{revisionfinal}{RGB}{0,0,0}
\newcommand{\oldfrag}[1]{}
\newcommand{\notefrag}[1]{}
\newcommand{\finalfrag}[1]{#1}

\newcommand{\origblock}[1]{}
\newcommand{\origitem}[1]{}
\newcommand{\noteblock}[1]{}

\fi

\newcommand{\Baocheng}[1]{}
\newcommand{\bwarn}[1]{}
\newcommand{\twarn}[1]{}

\usepackage{lineno}

\begin{document}

\date{}

\title{Probabilistically Checking Quantum Proofs, with Interaction}

\author{
Baocheng Sun\thanks{\texttt{baocheng.sun@weizmann.ac.il}. École Polytechnique Fédérale de Lausanne. Work done while at the Weizmann Institute of Science.
} \and
Thomas Vidick\thanks{\texttt{thomas.vidick@epfl.ch}. École Polytechnique Fédérale de Lausanne and Weizmann Institute of Science.
}}

\begin{titlepage}
  \maketitle
  \begin{abstract}

The model of interactive oracle proofs (IOP) generalizes the notion of probabilistically checkable proof (PCP), in which a static proof is verified probabilistically by querying a small number of bits, to the interactive setting: a polynomial-time verifier interacts with an unbounded prover, but is restricted to only reading a small number of bits, in total, from the messages sent by the prover. IOPs provide a relaxed setting in which to study local probabilistic verification. They have proved instrumental in devising efficient methods for verification through subsequent compilation into non-interactive or succinct protocols.

We study a quantum analogue of interactive oracle proofs (qIOP) in which the verifier and communication are both allowed to be quantum; yet the verifier is restricted to perform measurements only on a small number of qubits received from the prover. Our main result is a qIOP for any language in QMA, in which the total communication is polynomial but the verifier only reads a polylogarithmic number of qubits in total. The protocol has completeness parameter exponentially close to $1$ and soundness bounded away from $1$ by a constant. In the absence of a quantum PCP theorem, this provides the first information-theoretically sound local and robust characterization of QMA, albeit interactive.

Our protocol combines the use of a quantum locally testable code (LTC) with classical techniques, notably probabilistically checkable proofs of proximity (PCPP). We avoid the necessity for complex multi-qubit tests employed in other settings by leveraging the local indistinguishability property of the quantum LTC.

  \end{abstract}
  \thispagestyle{empty}
\end{titlepage}

{\hypersetup{linkcolor=black}
\tableofcontents
}
\newpage
\setcounter{page}{1}

\section{Introduction}

The PCP theorem shows that any problem which has efficiently verifiable proofs, has proofs that in addition can be verified with high confidence by probing only a small, in fact constant, number of bits. Based on the standard conjecture that $\QMA\neq\NP$ it is widely believed that there exist problems which admit efficiently verifiable \emph{quantum} proofs, but no classical proofs. Do all such problems have (quantum) local verifiers, which only probe a constant number of qubits from the proof? This is the content of the \emph{quantum PCP conjecture}, which has been open for more than two decades. The main difficulty, colloquially, is that due to the presence of entanglement information in a quantum proof can be encoded globally and is not necessarily accessible to local probes.

In this work we give a new characterization of $\QMA$ in terms of (quantum) efficient \emph{interactive} and local verification procedures. We study a recently introduced~\cite{sv-qiop} quantum analogue of \emph{interactive oracle proofs} (IOP) and show that every language in $\QMA$ has an efficient quantum interactive oracle proof (qIOP) where the verifier runs in quantum polynomial time and ``reads'' only a small (polylogarithmic, or even constant) number of qubits from the proof.\footnote{Here, we put quotes because the notion of ``reading'' a qubit can take different formalizations; we adopt a natural one which is described later on.} In addition, in our proof system the honest prover is relatively efficient and can be executed in quantum polynomial time given access to the right witness for the $\QMA$ language instance.

\subsection{Results}
\label{sec:results}

To introduce our results in more detail we first need to sketch the model of quantum interactive oracle proof (qIOP). Informally, a qIOP (for a language $L$) is an interactive protocol between a quantum polynomial-time verifier and a quantum unbounded-time verifier in which the verifier ``reads only a small number of qubits, in total, that were sent by the prover.'' Here and henceforth, we use ``small'' to refer to a quantity that is constant or at most polylogarithmic in the instance size. The protocol should satisfy the usual completeness and soundness conditions: for any instance that is in the language, there should exist a prover that is accepted with probability at least $2/3$;\footnote{For this case, we may in addition require the prover to be efficient given a witness. Whenever this holds, it should be stated explicitly.}  and for any instance that is not in the language, any prover, even computationally unbounded, should be accepted with probability at most $1/3$.

Of course, the subtlety in the definition lies in the precise formalization of the intuitive condition ``reads only a small number of qubits.'' Even for the case of a quantum probablistically checkable proof (qPCP), which corresponds to the non-interactive variant of a qIOP, the notion is not trivial to formalize and some work needs to be done to show that different natural formulations are equivalent (when they are); see e.g.\cite{aharonov2008u,buhrman2025quantum} or even~\cite{nn24}.

A formal definition of qIOP
appears in~\cite{sv-qiop}. We give an adaptation of the same definition in  Section~\ref{sec:qiop}, removing variations that are not needed for this paper; we refer the reader to that section for the formalities. For the purpose of the introduction we proceed by giving a semi-formal definition that is tailored to our protocol construction and yet can be seen to fit the formal definitions from Section~\ref{sec:qiop} in a straightforward manner.

Given a language $L$ and functions (of the input length) $\ell,q,r$, an $(\ell,q,r)$-qIOP for $L$ is given by the description of a quantum polynomial-time verifier which interacts with a quantum prover through a $(2r+1)$-message interaction. Without loss of generality, the first message comes from the prover
and consists of an $\ell_0$-qubit register $\sW$, which we think of as containing a suitably encoded quantum proof.
Each subsequent round $i=1,\ldots,r$ is a $2$-message interaction which proceeds as follows.
The verifier may perform three actions: measure up to $q_i$ qubits from the prover;
send $\sM_i$, a subset of qubits from $\sW$ that is deterministically chosen before the protocol begins, to the prover;
and send a classical message $\sC_i$, possibly depending on some private coin tosses, to the prover.
The prover applies an arbitrary unitary on all his registers.
He performs a measurement and returns a classical register $\sC'_i$ to the verifier.\footnote{The prover may be asked to return the register $\sM_i$, but in our protocol this is not needed by the verifier. Moreover, note that in general it cannot be enforced what qubits the prover sends back; only how many qubits are sent.}
(Let $\ell_i$ be the combined length, in qubits, of $\sM_i$, $\sC_i$ and $\sC'_i$ ; the parameters $\ell$ for the qIOP is defined as $\sum \ell_i$.)
At the last stage, i.e.\ after having received  $\sC'_r$ from the prover, the verifier flips some private coins and reads at most $q-\sum q_i$ bits from all classical registers he has in his possession (i.e. private classical registers, and all the $\sC'_i$ registers).
Based on the bits read, and any measurement outcomes and classical random bits generated by the verifier earlier in the protocol, the verifier makes his decision.

We note that this formulation precludes some adaptations that would be natural. For example, one could allow the verifier to adaptively determine which qubits to send back in each round. The verifier could also send qubits that were not directly obtained from the prover. The general model of qIOP naturally allows such operations ; because our protocol does not require them we omit them from the preceding discussion, for simplicity.

For the protocol to be considered a qIOP, it is required that $q$ is small; whereas $\ell$ is allowed to be an arbitrary polynomial. While in principle $r$ may also be polynomial, for us it is sufficient to take $r$ to be small as well (polylogarithmic).
Using this formalism our main result can be informally stated as follows.

\begin{theorem}[Main result, informal]
Every language in $\QMA$ has a $(\poly,\polylog,\polylog)$-qIOP.
In addition, the qIOP has completeness exponentially close to $1$ and the honest prover is efficient, i.e.\ it can be executed in quantum polynomial time provided it is initially given access to the right quantum witness.
\end{theorem}

We note a few desirable features of our protocol. Firstly, the verifier's additional classical bits sent to the prover are restricted to a single bit for all but the last message, which is polynomial (but could be made logarithmic using a stronger QMA-complete problem as starting point, see Section~\ref{sec:amplified-xz} below).

Secondly, the polylogarithmic query and round complexity are due to the sub-optimal parameters known for the best quantum locally testable code (qLTC) construction known~\cite{dlv24}, a key ingredient in our protocol. If there was a good qLTC construction, with positive rate, constant relative distance and soundness, and constant-weight parity checks, then the query and round complexity of our protocol would also both be improved to constant.

Thirdly, the total communication in the protocol is essentially quasi-linear, \emph{except} for the length of the (classical) PCPP proof that is sent by the prover to the verifier, which to the best of our knowledge may be super-linear, e.g.\ quadratic in length. If linear PCPPs were obtained then the total communication in our protocol would be linear, up to polylogarithmic factors. Moreover, we note that the only quantum communication consists of (i) the quantum witness sent from prover to verifier in the first round, and (ii) successive blocks of the same witness sent back from the verifier to the prover in each of the $r$ rounds of the protocol.

Nevertheless, we remark that the parameters as stated in the theorem may appear quite loose to a reader familiar with the literature on classical IOPs. Indeed, in the latter case there is a focus on the proof length (the parameter $\ell$ here) which is usually desired to be quasi-linear. This is because one already has access to the PCP theorem, which provides a polynomial-size proof, and the focus for IOPs is on strong efficiency. For us, there is no known quantum PCP and so any polynomial length proof is interesting. Furthermore, the size of the prover's first quantum message, in register $\sW$, is actually linear in the $\QMA$ proof length (for the specific language that we choose), while subsequent message registers sent from the (honest) prover to the verifier can be longer but are ``classical'' in the sense that they are obtained by performing a classical computation in superposition on subregisters $\sW_i$ returned by the verifier to the prover.

\subsection{Proof ideas}

We now discuss the proof of our result, namely the construction of our quantum IOP for QMA and the completeness and soundness arguments.

\subsubsection{Amplified $XZ$ Hamiltonians}
\label{sec:amplified-xz}

Our starting point is the $\QMA$-complete local Hamiltonian problem~\cite{ksv02}. For reasons that will become clear later on, it will be convenient for us to start with a local Hamiltonian $H$ that decomposes as an average of two local Hamiltonians $H=\frac{1}{2}(H^Z+H^X)$, one of which, $H^Z$, is diagonal in the computational basis and the other, $H^X$, in the Hadamard basis. The $\QMA$-completeness of such ``$XZ$ Hamiltonians'' was already shown in~\cite{biamonte2008realizable}.\footnote{In that reference such Hamiltonians, when in addition they are $2$-local, are called ``$ZZXX$ Hamiltonians.'' We use the terminology that was subsequently used in e.g.~\cite{grilo2019simple,mnz24}.}

The associated completeness and soundness parameters (expressed as the smallest eigenvalue of $H$ for YES and NO instances respectively), however, are only inverse polynomially separated. To obtain a qIOP with constant completeness-soundness gap this needs to be amplified before any protocol is attempted.\footnote{This is because e.g.\ sequential amplification of the protocol itself would require polynomially many repetitions, ruining any chances at keeping the query complexity small.} The usual tensoring amplification as used in e.g.~\cite{mnz24} leads to mixed terms (with components in both bases), which is inconvenient for us. Instead, we need to make two modifications.

First, we need to start with a complete problem that has the desired structure, in terms of separation of the local terms by $X$ or $Z$ basis, but such that in addition each local term is a projection, the completeness is exponentially close to $0$ (since we are talking about the minimum eigenvalue), and the soundness is inverse polynomially large. This will be needed for the subsequent amplification step. The result of~\cite{biamonte2008realizable} does not immediately give this. A close result is the $\QMA$-completeness proof of the Clifford Hamiltonian problem from~\cite{bjsw16}, which has the required parameters but is based on a richer set of bases than $X$ and $Z$. Inspired by their work, we developed a bespoke reduction from $\QMA$ to a local Hamiltonian $H=\frac{1}{2}(\Pi^Z+\Pi^X)$ where $\Pi^Z$ and $\Pi^X$ are normalized
averages of local projections diagonal in the computational and Hadamard basis respectively. The same reduction has been independently been obtained by Natarajan and Ma~\cite{mn25}, with similar motivation. Since their work appeared before ours, we directly use their result.\footnote{And, obviously, do not claim credit for it.}

We then amplify the completeness-soundness gap of $H$ while keeping the two-bases structure intact, e.g.\ we obtain $H^{amp}=\frac{1}{2}(\Pi^{Z,amp}+\Pi^{X,amp})$ where $\Pi^{Z,amp}$ and $\Pi^{X,amp}$ are normalized averages of non-local (but still efficiently described as tensor products of local terms, up to a subtraction of the identity)
projections diagonal in the $Z$ and $X$ basis respectively. Furthermore, the smallest eigenvalue of $H$ is exponentially close to
$0$ in the case of a positive instance, and larger than a universal constant in the case of a negative instance. This amplification is achieved using a simple method due to Anshu and explained in Section~\ref{sec:qma}. We note that the amplification leads to $\Pi^{Z,amp}$ and $\Pi^{X,amp}$ that are averages of an exponential number of non-local terms; yet a given non-local term can still be efficiently sampled and described. This leads directly to the length of the last classical message sent by the verifier in our protocol to be polynomial. It is possible to make that message logarithmic by using the more advanced, derandomized amplification technique in~\cite{bmvz25}. Since at the moment we do not know if the length of that message is important (see  Section~\ref{sec:open} on open questions), we do not include details on this modification.

\subsubsection{Trusted measurement provers}
\label{sec:trusted-prover}

Given an amplified $XZ$ Hamiltonian $H=\frac12(\Pi^Z+\Pi^X)$ on $B_q$ qubits provided as input, we warm up by describing a simple verification protocol in the case of a ``trusted measurement prover.'' This is a fictional prover (considered only for the purpose of this introduction) who holds an arbitrary quantum state $\ket{\varphi}_\sW$ of $M$ qubits and returns $M$-bit measurement outcomes of $\ket{\varphi}_\sW$ in the computational ($Z$) or Hadamard ($X$) basis whenever asked to do so.

Then the following simple, classical protocol would work. Since the verifier is only allowed to access a small number of bits of any outcome reported by the prover, and we do not assume any promise on the degree or locality of $H$ (so a small number of bit or phase flips can potentially affect the energy by a lot), it is sensible to require
$\ket{\varphi}_\sW$ to be an encoding of a ground state $\ket{\Gamma}$ of $H$ using a suitable quantum error correcting code. Thus let $C_q$ ($q$ for ``quantum'') be an $[M,B_q,d_q]$
quantum error correcting code with linear distance $d_q\propto M$, such that $C_q$ is furthermore a CSS code with parity check matrices $(H_X,H_Z)$. Let $\tE_q$ the unitary encoding map of $C_q$, and $\ket{\varphi}_W = \tE_q (\ket{\Gamma}_\sW\ket{0}_\sS)$ where $\ket{0}_\sS$ are ancilla qubits required for the encoding process.

Now suppose the verifier chooses a basis $W\in\{X,Z\}$ at random and requests measurement outcomes obtained from measuring $\ket{\varphi}_W$ in the basis $W$. The trusted measurement prover obliges, sending outcomes $w\in\{0,1\}^M$.
The verifier then samples a projection $\Pi_j^W$
in such a way that $\Pi^W = \E_{j} \Pi_j^W$. Since it is diagonal in basis $W$, the constraint $\Pi_j^W$ can be represented by a {boolean} function $f_j:\{0,1\}^{B_q}\to\{0,1\}$ (we suppress the dependence on $W$ for clarity) such that $\Pi_j^W = \sum_x f_j(x) S_W\proj{x}S_W^\dagger$, where $S_W=\Id^{\otimes B_q}$ if $W=Z$ and $S_W=H^{\otimes B_q}$ if $W=X$. Furthermore, we can represent the function $f_j$ as a ``logical'' function $f'_j:\{0,1\}^M\to \{0,1\}$ by using logical operators for the code $C_q$.\footnote{Somewhat informally, we map a string $x\in\{0,1\}^M$ to $x'\in \{0,1\}^{B_q}$ by evaluating the logical operators, i.e. $x'=Kx$ where the rows of $K$ represent the support of logical operators for $C_q$ in basis $W$. We refer to the definition of $\tilde{f}$ from $f$ in Section~\ref{sec:notation} for the precise methodology.} For reasons that will soon become clear, define $\tilde{f}_j = f'_j\circ \tE_W \circ \tD_W$, where $\tE_W$ and $\tD_W$ are the encoding and decoding maps associated with the linear code $C_W=\ker H_W$ respectively.\footnote{Note that $C_W$ does not necessarily have distance. The only property we will need is that the decoder $\tD_W$ decodes a small syndrome to an error whose size is not much bigger. This property is satisfied by the decoder in basis $W$ for the code we use, see Section~\ref{sec:intro-ltc} and Section~\ref{sec:decoding}.\label{fn:local}}

The goal of the verifier is to check that $f'_j(w)=0$. If $w\in C_W$, as it would be in case $\ket{\phi}_\sW$ is the intended witness, this is equivalent to checking that $\tilde{f}_j(w)=0$. So, it is enough for the verifier to check that $w\in C_W$ AND $\tilde{f}_j(w)=0$. Since it cannot do either task by itself (in particular, $\tilde{f}_j$ is nonlocal) it will enlist the prover's help.

To explain this first recall that a \emph{PCP of proximity} (PCPP) for an (efficiently computable) statement of the form ``$g(y)=0$'' is a proof $\pi$ that guarantees that $y$ is \emph{close} to a string $y'$ such that $g(y')=0$, and such that verifying $\pi$ only requires reading a constant number of bits of $y$ and $\pi$.

The prover is then asked to provide a PCPP for the statement: ``$\tilde{f}_j(w)=0$ AND $H_W w=0$''. Here, the role of $y$ is played by $w$, and the role of $g$ is played by the function $(1\oplus\tilde{f}_j(\cdot))\wedge 1_{H_W(\cdot)=0}$.\footnote{Note that in the actual protocol, the verification that $w$ is a codeword of the code $C_W$ is performed by running the local testing algorithm of $C_W$, which is locally testable (see Definition~\ref{def:qltc}). }
By the PCPP guarantee, if the PCPP proof is accepted with high probability then $w$ is \emph{close} to a string $w'$ such that $H_W w'=0$ and $\tilde{f}_j(w')=0$. We claim that it must be that $\tilde{f}_j(w)=0$ as well. This is because by the soundness property of the PCPP, we must have $w=w'+a$ where $a$ is a small error. Using that $\tD_W$ is a local decoder (see footnote~\ref{fn:local}), $a'=\tE_W \circ\tD_W(a)$ is a small element of $\Im(H_{W'}^T)$ (where $W'$ is the complementary basis to $W$). Thus for every logical operator, $K_i \cdot a'=0$ where $K_i$ is the support of the $i$th logical operator in basis $W$ (see Lemma~\ref{lem:q_noise}). It follows from the definition that $\tilde{f}_j(w)=\tilde{f}_j(w')$, as desired.  This concludes the analysis of the situation where the prover is a ``trusted measurement prover.''

\subsubsection{Verifying measurement outcomes}

Now of course the main difficulty remains: we need to devise a query-efficient procedure by which the verifier can determine that the prover correctly reports a string of measurement outcomes in the $X$ or $Z$ basis, where in either case the measurement has been performed on the same quantum state $\ket{\varphi}_\sW$ (which need not be pure).

Interactive procedures for verifying measurement outcomes in these bases have been developed in other settings, and form the cornerstone of delegated computation protocols. For example, in the two-prover setting it is by now well-understood that building on non-local games such as the CHSH or Magic Square games, and combining with techniques from classical locally testable codes, one can devise games that force the prover to measure in the required bases; see e.g.~\cite{Ji2022Quantum,chapman2023efficiently}. Similarly, Mahadev's delegated computation protocol~\cite{mahadev2018classical}, with a single quantum prover and a classical verifier, makes use of computational assumptions to essentially force the prover to make measurements in two mutually unbiased bases; see e.g.~\cite{gunn2025classical} for more on this. We also note that the two-prover multi-qubit tests have been ported to the single-prover, classical-verifier setting through compilation, based on quantum fully homomorphic encryption; see e.g.~\cite{nz23,mnz24}.

In our setting, we have access to a single quantum prover and no computational assumptions. However, we may leverage quantum communication. A recent paper~\cite{sv-qiop} introduces a direct analogue of the two-prover multi-qubit tests by using interaction and hiding the verifier's choice of test in entanglement between the verifier and prover.
However, in that paper they were only able to achieve a qIOP with exponential communication.
Our model and theirs are closely related, yet not entirely comparable.

Here we use a different idea, which leverages the properties of quantum error correcting codes and quantum locally testable codes.

\subsubsection{Leveraging local indistinguishability}
\label{sec:loc-ind}

We first observe that there is a very classic single-qubit measurement test that leverages quantum communication. The verifier prepares an EPR pair and sends one qubit to the prover. The verifier asks the prover to measure their qubit in a randomly chosen basis $W\in\{X,Z\}$ and report the outcome $x$. The verifier measures their own half-EPR pair to obtain an outcome $y$, and accepts if and only if $x=y$. It is easily verified that the only possibility for winning perfectly in this game is for the prover to measure the qubit he received from the verifier in the correct basis.

The main limitation of this protocol is that the qubit measured is a half-EPR pair; whereas what we really want is a measurement of the quantum witness. A similar difficulty in the two-prover setting was resolved by the use of teleportation in~\cite{grilo2019simple}. Here, we employ a different idea, which is inspired by the local indistinguishability property of quantum error-correcting codes. Recall that in Section~\ref{sec:trusted-prover} we already required the witness to be encoded using a quantum code $C_q$.
If the code has distance $d_q$, then it is well-known that the reduced density matrix $\rho_S$ of any codeword on a subset $S$ of $|S|<d_q$ qubits must be independent of the specific codeword. Suppose for simplicity that $\rho_S$ is totally mixed. Then this means that, up to unitary rotations, the qubits in $S$ and the qubits in $\overline{S}$ are equivalent to a number of EPR pairs! This suggests that a suitable adaptation of the simple test exposed in the previous paragraph could lead to certifying the prover's measurements on any sufficiently small portion of the code.

Concretely, we consider the following  procedure. Divide the $M$ qubits of an encoded state into $r$
blocks such that each block has less than $d_q$ qubits. For any $i\in\{1,\ldots,r\}$, the verifier sends to the prover the qubits in the $i$th block only. The verifier requests a measurement in basis $W$ of those qubits, leading to outcomes $x$. The prover reports $x$. The verifier measures all other qubits in basis $W$, leading to an outcome $x'$. The verifier checks that together, $(x,x')$ satisfy the code stabilizers in basis $W$, i.e.\ $H_W(x,x')=0$.

It is not hard to show that a prover succeeding with probability $1$ in this test must indeed perform honest measurements on the qubits of the codeword that it is given. This idea forms the basis for our measurement extraction procedure, described in full in Section~\ref{sec:gme}. However, we need to perform a number of modifications that introduce additional technical difficulties.

Firstly, the description we gave requires the verifier to perform a large measurement on the witness, which is not allowed by the qIOP model. To alleviate this we rely on a quantum code that not only has low-weight parity checks but furthermore is locally testable. This ensures that the verifier (together with the prover's outcomes) will be able to assess the distance of the quantum witness from the quantum code space by measuring only a small number of qubits. Concretely we use the quantum locally testable (qLTC) codes based on cubical complexes from~\cite{dlv24}. Such codes have good but not quite near-optimal parameters, i.e.\ for code length $M$ they have linear rate, distance $\Omega(M/\poly\log M)$, soundness $\Omega(1/\poly\log M)$ and constant-weight parity checks.\footnote{Some tradeoffs are possible between these parameters~\cite{wills2024tradeoff}; for clarity we do not explore these and stay with the ``basic'' construction.}

Secondly, in order to estimate the energy of the amplified $XZ$ Hamiltonian $H$ we of course need to certify (and obtain) measurements on the entire witness, not only a small block of qubits. This is necessary because by the same principle of  local indistinguishability, no information about the witness could be extracted by measuring only a small number of qubits.
To address this we introduce an interactive procedure in which measurements are performed by the prover and only read and certified by the verifier when needed.
We describe both aspects next.

\subsubsection{Interactive measurement extraction using locally testable codes}
\label{sec:intro-ltc}

We sketch the procedure by which the verifier is able to iteratively extract trusted measurement outcomes in basis $W$ from the prover. The precise procedure and analysis are given in Section~\ref{sec:gme}. The quantum witness sent by the prover to the verifier in the initial message is divided into $r=M/m$ blocks of qubits, where $m \ll d_q$. The protocol then proceeds in $r$ rounds (note that eventually we will have $r=\poly\log(M)$). At the $i$th round, the $m$ qubits in the $i$th block are sent to the prover. The prover is asked to measure them in basis $W$ to obtain a string $x\in\{0,1\}^m$. The prover then reports this string $x$ to the verifier.

In a randomly chosen round, the verifier samples a random basis-$W$ stabilizer generator $g_W\in\{I,W\}^{\otimes M}$ for the quantum code $C$ and verifies it as follows. If a qubit in the support of $g_W$ has never been sent back to the prover, then the verifier measures it directly. If however the qubit has already been sent back to the prover, and thus measured in basis $W$ then the verifier uses the measured value $x_j$ returned by the prover instead.

To analyze this procedure we need to take into account that in later rounds, most of the qubits of the codeword have previously been acted on by the prover, and thus the local indistinguishability principle no longer applies. Nevertheless we are able to leverage the LTC property, which holds separately for both classical codes that underlie the CSS quantum LTC code, to argue that any deviation from the prover in the complementary basis $W'$ (i.e. $X$ errors for an extraction in basis $Z$; these are the only errors that would affect a measurement in basis $W$) must have bounded weight. Using the linear distance property of the quantum code $C$, bounded weight errors do not affect logical operations. This guarantees that measurement outcomes are reported correctly, as desired.

To conclude it remains to bring together the two components, the constraint verification procedure described in Section~\ref{sec:trusted-prover} and the measurement extraction described here.

\subsubsection{Protocol summary}

We now summarize the structure of our protocol, from which it will be evident that the protocol has the parameters announced in Section~\ref{sec:results}. We describe the protocol in the case of a YES instance.

Let $H=\E_{W\in\{X,Z\}} \E_{j} \Pi^W_j$ be an instance of the amplified $XZ$ Hamiltonian problem described in Section~\ref{sec:amplified-xz}, let $\ket{\Gamma}$ be a ground state of $H$ and $\ket{\varphi}_\sW=E_q(\ket{\Gamma}\ket{0}_\sS)$, where $E_q$ is the encoding unitary for a CSS code $C=(H_X,H_Z)$ with  distance $d_q$ that is locally testable. Let $M$ be the code length. Divide $M$ into $r$ chunks of $m$ qubits where $m  \ll d_q$.

The protocol has two phases. The first phase is called the \emph{measurement extraction} phase, and the second phase is the \emph{constraint verification} phase.
We first describe the  measurement extraction phase. It consists of $1+2r$ messages, as follows:
\begin{enumerate}
\item The prover sends an $M$-qubit register $\sW$ to the verifier. This register is supposed to contain $\ket{\varphi}_W$.
\item The verifier samples a basis $W\in\{X,Z\}$ uniformly at random.
\item For $i=1,\ldots,r=M/m$ do the following:
\begin{enumerate}
\item The verifier sends qubits in the $i$th block of $\sW$ to the prover. If $i=1$, the verifier also sends the basis $W$ to the prover.
\item The prover  measures this block in basis $W$ to obtain outcomes $x_i\in\{0,1\}^m$.
The prover returns $x_i$.
\item In one of the rounds, chosen at random, the verifier executes the qLTC local tester for basis $W$.
Whenever a query lands on a qubit that has previously been measured by the prover, the verifier uses $x_j$ as the measurement outcome. Whenever the qubit has never been sent to the prover, the verifier measures it directly. The verifier rejects if the qLTC tester rejects.
\end{enumerate}
\end{enumerate}
At the end of this phase, the verifier has in his possession a string $x$ that is supposed to contain measurement outcomes corresponding to measuring the encoded witness $\ket{\varphi}_\sW$ in basis $W$. Our analysis (see Lemma~\ref{lem:gme}) shows that, indeed, any (efficiently computable) ``logical'' function $\tilde{f}$ of $x$ computed at the end of the phase is consistent with the evaluation of the ``honest'' observable associated with $\tilde{f}$ on $\ket{\varphi}_\sW$ at the beginning of the phase.

In the second and last phase, the \emph{constraint verification} phase, the verifier interacts with the prover one more time to evaluate a non-local Hamiltonian constraint $\Pi_j^W$ on $x=(x_i)_{i\in [r]}$.
This phase consists of only $2$ messages, as follows:
\begin{enumerate}
\item The verifier selects a random constraint $\Pi^W_j$ from the Hamiltonian $H$. Let $f_j:\{0,1\}^{B_q}\to\{0,1\}$ be the associated function, as in Section~\ref{sec:trusted-prover}. The verifier sends $j$ to the prover.
\item The prover sends the value $b=f(KE_W\tD_W(x))$ together with a PCPP proof that $x$ is a valid codeword of $C_W$ such that $f(KE_W\tD_W(x))=b$. Here, $\tD_W$ is the $W$-basis decoder for the qLTC, $E_W$ is the encoder for the code $C_W=\ker H_W$, and $K$ is a matrix of logical operators (in rows) for basis $W$. (See Section~\ref{sec:trusted-prover} for why this representation is used.)
\item The verifier accepts if and only if the PCPP proof is accepted and $b=0$.
\end{enumerate}
The analysis of this phase, given the results of the previous phase, is fairly straightforward and follows the sketch given in Section~\ref{sec:trusted-prover}.

\subsection{Related work}

Previous works have proved or conjectured characterizations of $\QMA$ in terms of efficient, ``local'' verification in other settings. The \emph{quantum games PCP} conjecture~\cite{nn24} considers the case of a classical verifier who exchanges $O(\log n)$-length messages with two provers sharing entanglement, where $n$ is the instance size. While this conjecture remains open, in contrast to the quantum PCP conjecture mentioned earlier there has been notable progress on it in the past years and its resolution appears close (at least to the present authors). In particular, the development of so-called multi-qubit tests as e.g.\ the Pauli Braiding test of~\cite{nv16,natarajan2019neexp,ji2021mip}, which emphasize the role of tensors product observables involving separate $X$ and $Z$ observables only, influenced this work.

In a different direction, previous works have considered the case of a classical verifier running in polylogarithmic time (given the right input access model) and shown that languages in $\QMA$ have interactive verification protocols of this sort, under computational assumptions~\cite{zhang2021succinct,bkl+22,nz23,mnz24}. We only consider the information-theoretic setting (but see the open problems subsection below).

Our work is arguably more related to early information-theoretic protocols for quantum verification with a small-size quantum verifier~\cite{aharonov2017interactive,fitzsimons2017unconditionally,fitzsimons2018post}. The main difference with our work is that, in all these results, the total quantum effort of the verifier, aggregated across rounds, remains polynomial. In our paper, the total quantum effort of the verifier is polylogarithmic. The difference is substantial; and indeed, while the cited works do not rely on further structural results than the quantum circuit model and Kitaev's reduction showing that the local Hamiltonian problem is $\QMA$-hard, our results require a tailored, amplified $\QMA$-hard problem to start with; classical probabilistic proof verification technique (in the form of PCPPs); as well as recent advances in quantum codes and in particular the design of quantum locally testable codes.

\subsection{Open problems}
\label{sec:open}

In the classical literature on delegated computation, interactive oracle proofs are a step towards more ``practical'' primitives. The next step is achieved by removing interaction and/or achieving succinctness by replacing the long, mostly unused messages by efficient commitment schemes such as a Merkle tree. It is thus natural that the next step in our work would be to seek methods for ``compiling'' our quantum interactive PCP into a succinct and/or non-interactive scheme, based e.g.\ on the quantum random oracle model. A few features of the protocol make this step, in particular the non-interactive part, non-trivial. Notably, our protocol is private-coin (for example, qLTC stabilizer checks are done privately), and moreover involves sending quantum messages from the verifier to the prover; such messages are not easily removed by the standard Fiat-Shamir heuristic. Some approaches exist to Merkle tree-like constructions in the quantum setting~\cite{chen2024quantum,gunn2023commitments} and it would be natural to try to use those first in order to achieve succinctness. We leave such attempts for future work.

Another potential direction would aim to connect our work to ``compiled'' protocols as in~\cite{mnz24}, which are interactive protocols with a classical succinct verifier, obtained at the cost of relatively heavy computational assumptions such as quantum fully homomorphic encryption. It is possible that our protocol could be more directly compiled into one with a classical verifier, ideally while using weaker computational assumptions.

More immediate directions aim to reduce the complexity of our interactive protocol along different axes. It is possible that the protocol can be made public-coin, as the most important random choices of the verifier, the basis in which to extract measurement outcomes and the term of the Hamiltonian that is estimated in the energy check, are already done publicly. Other random choices, such as which round to perform a test in or which code stabilizer is measured, are made privately in our formulation; yet it appears that it would be straightforward to make them publicly.
A more difficult direction is to remove quantum communication from the verifier to the prover, which could aid compilation. One could also focus on the total communication length. At the moment the most important components, such as the qLTC-encoded witness, are quasi-linear, but other components, such as the PCPP,  remain super-linear, to the best of the authors' knowledge.

Finally, while our work focuses on establishing a lower bound on the power of qIOPs, i.e.\ on designing a concrete protocol, it is of course interesting to study potential limitations on the model. As one puts restrictions on the communication, the locality, etc.\ when does one cross the boundary between QMA and NP (or, say, QCMA)?

\paragraph{Organization of the paper.} In Section~\ref{sec:prelim} we review some mostly well-known facts about classical and quantum error-correcting codes, and classical probabilistically checkable proofs of proximity (PCPP). The following Section~\ref{sec:notation} summarizes important notation used throughout, and should be referred to whenever a notation or symbol are not clear. Section~\ref{sec:qiop} and Section~\ref{sec:qma} contain further, less standard background; the first on quantum IOPs and the second on the QMA-complete problem that we use as starting point. Section~\ref{sec:css} establishes important lemmas on the structure of qLTC near-codestates and various classical computations on them.
In Section~\ref{sec:pcp-proxim} we detail how PCPs of proximity are computed, coherently, on codestates; the results in that section are used later on (it can be skipped at first read).
The last three sections contain the bulk of the analysis: Section~\ref{sec:gme} gives the guarantees we can obtain of the measurement extraction procedure; Section~\ref{sec:measurement} on the procedure for computing functions of the measurement outcomes; and Section~\ref{sec:energy} on the method for evaluating the energy of the prover-provided witness with respect to the input Hamiltonian.

\paragraph{Acknowledgments.}
This work is supported by ERC/SERI grant VerNisQDevS.
\section{Preliminaries}
\label{sec:prelim}

We collect mostly standard material on quantum error correcting codes, locally testable codes (LTCs), and probabilistically checkable proofs of proximity (PCPPs).

\subsection{Classical locally testable codes}

We now formally define classical locally testable codes, which guarantee that proximity to the codespace can be efficiently checked using a small number of queries.
\begin{definition}[Classical locally testable codes]
\label{def:ltc}
Let $C = \ker(H)$ be a classical code, where $H\in\mathbb{F}_2^{m\times n}$ is a parity-check matrix. We say that $C$ is \emph{locally testable} with soundness $\rho$ if for all $x \in \mathbb{F}_2^n$, it holds that
\[
\frac{1}{m} |Hx| \geq \rho \frac{d_H(x, C)}{n}\;,
\]
where $|x|$ denotes the Hamming weight (number of nonzero coordinates) of $x$ and $d_H(x, C)$ denotes the Hamming distance between $x$ and the code $C$.
\end{definition}

We define the local tester associated with a classical locally testable code. This tester checks whether a given string is a valid codeword by examining a small number of randomly selected parity constraints. Note that the tester depends on the choice of a parity check matrix for the code.

\begin{mdframed}
\begin{definition}[Local tester for a classical LTC]
\label{def:local-tester}
Let $c \in \mathbb{Z}_+$. Let $C = \ker(H)$ be a classical locally testable code, where $H \in \mathbb{F}_2^{m \times n}$ is a parity-check matrix. Define the local tester $\texttt{T}$ for the code $C$ as the following procedure:
\begin{enumerate}
    \item Given input $w \in \mathbb{F}_2^n$,
    \item Randomly sample $i_1, \ldots, i_c \in_R [m]$,
    \item Compute
    \begin{align*}
        r =
        \begin{bmatrix}
    H_{i_1} \\
    \vdots \\
    H_{i_c}
    \end{bmatrix}
        \cdot w \in \mathbb{F}_2^c,
    \end{align*}
    \item Reject if \( \bigvee_{j \in [c]} r_j = 1 \) (i.e., any parity constraint is violated).
\end{enumerate}
\end{definition}
\end{mdframed}

We now present a method to amplify the tester of any locally testable code, for ease of application:

\begin{lemma}
\label{lem:ltc_amp}
Let $C = \ker(H)$ be a locally testable code with soundness $\rho$, where each row of $H$ has weight at most $R$.
Let $\kappa > \rho $.

Then, there exists a code tester for $C$ with query complexity $Q$ such that for all $x \in \mathbb{F}_2^n$,
\[
\Pr[\text{verifier rejects } x] \geq 1 - \left(1 - \rho \cdot \frac{d_H(x, C)}{n}\right)^{Q/R} \;.
\]

Furthermore, there exists a code tester for $C$ with query complexity
\[
Q = R \cdot \frac{\kappa}{\rho} \cdot \frac{\log(1/\varepsilon)}{1 - \varepsilon}
\]
such that for all $x \in \mathbb{F}_2^n$, it holds that
\[
\Pr[\text{verifier rejects } x] \geq \min\left\{\kappa \cdot \frac{d_H(x, C)}{n},\, 1 - \varepsilon \right\} \;.
\]
\end{lemma}
\begin{proof}
By independently sampling a row of $H$ a total of $N = Q/R$ times and checking whether the codeword satisfies the corresponding parity check, we obtain a code tester that rejects a near-codeword with the following probability:
\begin{align*}
\Pr[\text{verifier rejects}]
&= 1 - \left(1 - \frac{1}{m} |Hx| \right)^{N} \\
&\ge 1 - \left(1 - \rho \frac{d_H(x, C)}{n}\right)^{N}\;.
\end{align*}

By taking $N = \frac{\kappa}{\rho} \cdot \frac{\log(1/\varepsilon)}{1 - \varepsilon}$, we have
\begin{align*}
1 - \left(1 - \rho \frac{d_H(x, C)}{n}\right)^{N}
\ge 1 - \exp\left(-\rho \frac{d_H(x, C)}{n} N\right)\;.
\end{align*}

If $\rho \frac{d_H(x, C)}{n} N \ge \log(1/\varepsilon)$, then
\begin{align*}
1 - \exp\left(-\rho \frac{d_H(x, C)}{n} N\right) \ge 1 - \varepsilon\;.
\end{align*}

If $\rho \frac{d_H(x, C)}{n} N \le \log(1/\varepsilon)$, note that for $x \in [0, a]$,
\begin{align*}
\exp(-x) \le 1 - \frac{1 - \exp(-a)}{a}x.
\end{align*}
Then,
\begin{align*}
1 - \exp\left(-\rho \frac{d_H(x, C)}{n} N\right)
&\ge 1 - \left(1 - \frac{1 - \varepsilon}{\log(1/\varepsilon)} \cdot \rho \frac{d_H(x, C)}{n} N \right) \\
&= \kappa \frac{d_H(x, C)}{n}.
\end{align*}

Therefore, in all cases,
\begin{align*}
\Pr[\text{verifier rejects}]
\ge \min\left\{ \kappa \frac{d_H(x, C)}{n},\, 1 - \varepsilon \right\}.
\end{align*}
\end{proof}

We introduce a quantum analogue of the local tester. This operator allows the verifier to coherently verify proximity of a string to the code space using a classical local test.

\begin{definition}[Local tester operator]
\label{def:tester}
Let $\texttt{T}$ be the local tester algorithm for a classical LTC, as defined in Definition~\ref{def:local-tester}. Let $l_r$ denote the number of random bits used by $\texttt{T}$, and let $q$ denote the number of bits of $q$ that are queried by the tester. For $r \in \mathbb{F}_2^{l_r}$, let $\texttt{T}(w; r)$ denote the output of $\texttt{T}$ when the randomness is $r$ and the query results corresponding to that randomness are given by $w$.

Then, for $a \in \mathbb{F}_2$, $w \in \mathbb{F}_2^q$, and $r \in \mathbb{F}_2^{l_r}$, define the local tester operator as
\begin{align*}
L_r(\ket{a} \ket{w} ) = \ket{a \oplus \texttt{T}(w; r)} \ket{w}\;.
\end{align*}
\end{definition}

\subsection{Quantum CSS codes}

We first introduce standard notation for Pauli operators.

\begin{definition}[Pauli Operators]
\label{def:pauli}
The Pauli operators are defined as
\[
X = \begin{pmatrix} 0 & 1 \\ 1 & 0 \end{pmatrix}, \quad
Z = \begin{pmatrix} 1 & 0 \\ 0 & -1 \end{pmatrix}.
\]
\end{definition}

We introduce a convenient notation for applying a single-qubit operator selectively across multiple qubits, according to a binary string.
\begin{definition}[Tensor product of single-qubit operators]
\label{def:tensor_op}
Let $e \in \mathbb{F}_2^n$ and let $W$ be a single-qubit operator. Define
\begin{align*}
W(e) = W^{e_1} \otimes W^{e_2} \otimes \cdots \otimes W^{e_n},
\end{align*}
that is, $W(e)$ applies $W$ at each position $i$ where $e_i = 1$ and applies the identity operator otherwise.
\end{definition}

Next, we define CSS codes.

\begin{definition}[CSS Code]
\label{def:css}

An $[n, k, d]$ quantum CSS code on $n$ qubits is a subspace $C \subseteq \mathcal{H} = (\mathbb{C}^2)^{\otimes n}$.
It is defined by a pair of parity-check matrices $H_X, H_Z$ such that $\operatorname{Im}(H_X^\top) \perp \operatorname{Im}(H_Z^\top)$.
We denote the code as $C = CSS(H_X, H_Z)$.
Explicitly, the code subspace is given by
\[
CSS(H_X, H_Z) = \left\{ \ket{\psi} \in (\mathbb{C}^2)^{\otimes n} : X(x) \ket{\psi} = \ket{\psi} \, \forall x \in H_X, \, Z(z) \ket{\psi} = \ket{\psi} \, \forall z \in H_Z \right\},
\]
which can also be written as
\[
CSS(H_X, H_Z) = \text{Span}\Big\{ \frac{1}{\sqrt{|H_X|}} \sum_{x \in H_X} \ket{z + x} : z \in \ker(H_Z) \Big\}\;.
\]
The code encodes $k = \log_2\left( |\ker(H_X) / \Im(H_Z^T)| \right)$ logical qubits and has distance
\[
d = \min\big\{ |w| : w \in (\ker(H_X) \setminus \Im(H_Z^T)) \cup (\ker(H_Z) \setminus \Im(H_X^T)) \big\}\;.
\]
\end{definition}

Finally, we define the distance of the $W$-code modulo the image of the parity-check matrix of the complementary code.

\begin{definition}
\label{def:dist_mod}
Let $C$ be a quantum CSS code and $H_X$, $H_Z$ the associated parity check matrices, as in Definition~ \ref{def:css}. Define
\begin{align*}
d_{q,X} =&~ \min\big\{ |w| : w \in \ker(H_X) \setminus \Im(H_Z^T) \big\}\;,\\
d_{q,Z} =&~ \min\big\{ |w| : w \in \ker(H_Z) \setminus \Im(H_X^T) \big\}\;.
\end{align*}
\end{definition}

We now state a lemma that guarantees the existence of a natural basis of logical Pauli operators in CSS codes.

\begin{lemma}
\label{lem:logical_qltc}
Let $C$ be an $[M, B_q, d_q]$ CSS code. Then there exist $F, G \in \mathbb{F}_2^{B_q \times M}$ such that $$X(F_1), \ldots, X(F_{B_q})$$ are $B_q$ linearly independent logical operators composed only of $X$ and identity operators, and $$Z(G_1), \ldots, Z(G_{B_q})$$ are $B_q$ linearly independent logical operators composed only of $Z$ and identity operators. Moreover, for each $i, j \in \{1, \ldots, B_q\}$, these operators satisfy
\[
X(F_i) Z(G_j) = (-1)^{1(i=j)} Z(G_j) X(F_i).
\]
Finally, it holds that
\begin{equation}\label{eq:kfg}
F H_Z^T = G H_X^T = 0 \;.
\end{equation}
\end{lemma}

\begin{proof}
Eq.~\eqref{eq:kfg} follows since logical operators commute with the stabilizers.
\end{proof}

Finally we define appropriate unitary encoding and decoding maps. Regarding encoding, we provide a unitary encoding function based on the logical operators defined in Lemma~\ref{lem:logical_qltc}.

\begin{definition}[Encoding of a CSS code]
\label{def:e_qltc}
Let $C$ be an $[M, B_q, d_q]$ CSS code.
Let $ H_X,H_Z$ be defined as in Definition \ref{def:css}.
Let $G$ be defined as in Lemma~\ref{lem:logical_qltc}.

For each $a \in \F_2^{B_q}$, let
\begin{align*}
\ket{\phi_a} =  \prod_{j\in [k]} \frac{I+ (-1)^{a_j} Z(G_j)}{2} \cdot \prod_{j}\frac{I+ X((H_X)_j)}{2}\cdot \prod_{j}\frac{I+ Z((H_Z)_j)}{2} \ket{+}^{\otimes M}.
\end{align*}
Let the following unitary $E_q$ be the encoding unitary of $C$:
\begin{align*}
E_q(\ket{a}\ket{0}^{\otimes (M-B_q)})=&~ \frac{\ket{\phi_a}}{\|\ket{\phi_a}\|}\; .
\end{align*}
For convenience, we also write $E_q(\ket{\phi})$ to denote $E_q(\ket{\phi}\ket{0}^{\otimes (M-B_q)})$.
\end{definition}

\begin{remark}
Note that the $\ket{\phi_a}$ are orthogonal for different $a$, and $\|\ket{\phi_a}\|>0$. Thus $E_q$ is a well-defined unitary.
\end{remark}

The following decomposition of arbitrary states as linear combinations of codewords with Pauli errors will be convenient for our analysis.

\begin{lemma}[Stabilizer decomposition, \cite{got97}]
\label{lem:stabilizer-decomposition}
Let $C=CSS(H_X, H_Z)$ be a stabilizer code on $n$ qubits.
Let
\begin{align*}
\mathcal{S} = \Big\{ (a, b) \,\Big|\, a, b \in \mathbb{F}_2^n,\, \exists r,s \text{ s.t. } & |a \vee b| \text{ is minimized, and } (a, b) \text{ is lex smallest }\\
&~\text{ among all } (a', b') \text{ with } H_X b' = r \text{ and } H_Z a' = s \Big\}\;.
\end{align*}
Then any $n$-qubit state $\ket{\psi}$ has a decomposition of the form
\[
\ket{\psi} = \sum_{(a,b)\in \mathcal{S}} \alpha_{a,b} X(a)Z(b)  \ket{\eta_{a,b}}\;,
\]
where $\|\ket{\eta_{a,b}}\|=1$, $\ket{\eta_{a,b}} \in C$, the terms $X(a)Z(b)  \ket{\eta_{a,b}}$ are pairwise othogonal, and $\sum_{a,b}|\alpha_{a,b}|^2=1$.
\end{lemma}

\begin{proof}
Any $n$-qubit state $\ket{\psi}$ decomposes as
\[\ket{\psi}= \sum_{r,s} \alpha_{r,s}\ket{\psi_{r,s}}\;,\]
where $\ket{\psi_{r,s}}$ is a joint eigenvector of all $X(x)$, $x\in H_X$ and $Z(z)$, $z\in H_Z$ whose associated eigenvalues are given by the vectors $r,s$. This is because these eigenspaces form a direct sum decomposition of the entire space. For any such $r,s$, if we let $(a,b)$ be the pair associated to it as in the definition of $\mathcal{S}$, then $\ket{\psi_{r,s}}= X(a)Z(b)  \ket{\eta_{a,b}}$ for some codeword $\ket{\eta_{a,b}}$.
\end{proof}

We then define the weight of a Pauli error.

\begin{definition}[Weight of an error in stabilizer codes]
\label{def:weight-error-stabilizer}
Let $C=CSS(H_X, H_Z)$ be a stabilizer code on $n$ qubits. For $ a,b\in\F_2^n$, let $ E = X(a)Z(b)$. Define the weight of $E$ modulo the centralizer of the code generating set as:
\[
{wt}_{\text{cent}}(E) = \min_{\substack{a', b' \in \F_2^n,\\  H_X b'=0, H_Z a' = 0}} \big\{\ {wt}(X(a')Z(b') E)\ \big\}\;.
\]

\end{definition}
\begin{remark}
Note that, for $(a,b)\in {\cal S}$, ${wt}_{\text{cent}}(X(a)Z(b)) = |a\vee b| $.
\end{remark}

We end with the decoding operator for CSS codes, which will be used to map an encoded quantum state back to its logical components and syndrome information. Note that this decoding map, as defined, is not necessarily efficient; and we will only use it for purposes of analysis. In Section~\ref{sec:decoding}, we discuss efficient maps associated with more restricted forms of decoding.

\begin{definition}[Decoding qLTC code]
\label{def:d_qltc}
Let $ H_X,H_Z$ be defined as in Definition \ref{def:qltc}.
Let $ {\cal S}$ be defined as in Lemma \ref{lem:stabilizer-decomposition}.
Any $n$-qubit quantum state $\ket{\psi}$ decomposes as
\[
\ket{\psi} = \sum_{(a,b)\in {\cal S}} \alpha_{a,b} X(a)Z(b)  \ket{\eta_{a,b}},
\]
where $\|\ket{\eta_{a,b}}\|=1$ and $\ket{\eta_{a,b}} \in C$.
Let $D_q$ be the following unitary:
\begin{align*}
D_q(\ket{\psi}) = \sum_{(a,b)\in {\cal S}} \alpha_{a,b}  \ket{\eta'_{a,b}} \ket{H_Z a} \ket{H_X b},
\end{align*}
where $\eta_{a,b} = E(\eta'_{a,b})$ and $E$ is the qLTC encoding of $\eta'_{a,b}$.
\end{definition}
\begin{remark}
By Lemma~\ref{lem:stabilizer-decomposition}, $D_q$ is well-defined and is indeed a unitary. Note that $D_q$ does not necessarily need to be efficient.
\end{remark}

\subsection{Quantum locally testable codes}

We then introduce the standard definition of quantum locally testable codes.

\begin{definition}[Local testability for quantum CSS codes, \cite{dlv24}]
\label{def:qltc}
Let $C = \mathrm{CSS}(H_X, H_Z)$ be a quantum CSS code. We say that $C$ is \emph{locally testable} with soundness $\rho$ if both $C_X = \ker(H_X)$ and $C_Z = \ker(H_Z)$ are classical locally testable codes with soundness at least $\rho$.
\end{definition}

Next, we present a known construction of good quantum locally testable codes.

\begin{theorem}[Quantum locally testable codes, \cite{dlv24}]
\label{thm:qltc}
There exists a polynomial-time algorithm that, given an integer $N$ (in unary), returns explicit descriptions of parity-check matrices $H_X$ and $H_Z$ on $\Theta(N)$ qubits such that the following properties hold:
\begin{enumerate}
    \item The weight of each row of $H_X$ and $H_Z$ is $\Theta(1)$;
    \item The code $C = \mathrm{CSS}(H_X, H_Z)$ has dimension $\Theta(N)$ and distance $\Theta(\frac{N}{(\log N)^3})$;
    \item $C$ is a quantum locally testable code (qLTC) with soundness parameter $\rho = \Omega(\frac{1}{(\log N)^3})$.
\end{enumerate}
\end{theorem}

\subsection{Decoding the $X$- and $Z$-codes of a CSS code}
\label{sec:decoding}

As stated in Theorem~\ref{thm:qltc}, we use the quantum LTC code based on cubical complexes from~\cite{dlv24}. This code has an efficient decoder that generalizes the small-set bit-flip decoder, as adapted to quantum LDPC code construction in~\cite{dinur2023good} and extended to the qLTC construction from~\cite{dlv24} in~\cite{nguyen2025quantum}. As stated in~\cite[Theorem 3.3]{nguyen2025quantum} the decoder runs in linear time and correctly decodes errors up to weight $\Theta(N/\poly\log(N))$, where $N$ is the blocklength of the code.

For this work we need an additional property of the decoder, which is that if $e\in\{0,1\}^M$ is any ``error'' such that the associated syndrome in basis $W$, for $W\in\{X,Z\}$, is small, i.e. $\sigma=H_W e$ has small Hamming weight, then the decoder on input $\sigma$ (deterministically) returns $e'$ such that
\begin{equation}\label{eq:norm-red}
    |e+e'| \leq \poly\log(M)\cdot |\sigma|\quad\text{and}\quad H_W e' = \sigma\;.
\end{equation}
Note that in general this property is not automatic, as an arbitrary decoder only needs to guarantee that $e'$ has the correct syndrome $\sigma$ ; but $e'$ may differ from $e$ by an arbitrary element of $\Im (H_{W'}^T)$, where $W'$ is the complementary basis to $W$. Such elements may a priori have arbitrary weight.

However, the decoder from~\cite{nguyen2025quantum} is a local decoder that progressively constructs an error $e'$ such that $H_We'=\sigma$ by finding $e_1,\ldots$, such that $|\sigma+H_W(e_1)|<\sigma$, etc. until no progress is found and one sets $e'=e_1+\cdots$. Therefore, the number of iterations of the decoder is at most $|\sigma|$ and furthermore at each step only a ``local'' error is added, i.e.\ $|e_i|=O(\poly\log M)$ for each $i$. This is evident from the decoder description in Algorithm 3 and Algorithm 4 in~\cite{nguyen2025quantum} and establishes the desired property~\eqref{eq:norm-red}. We record it in the following fact:

\begin{fact}\label{fact:decoder}
Let $C_q$ be the qLTC family from Theorem~\ref{thm:qltc}. Then $C_q$ admits an efficient (in fact, linear time) decoder $\tD_q$ for errors of weight up to $\eta_q d_q$, where $d_q$ is the quantum distance of the code and $\eta_q$ is a constant.
Furthermore, let $\tD_{W}$ be the decoder in basis $W\in\{X,Z\}$. We take this to be the map $\texttt{D}_W: \F_2^{M} \rightarrow \F_2^{B_W}$ that to a noisy vector returns the ``closest'' element of the code $H_W$
\footnote{Sometimes, the decoder is taken to be a map $\tD'_W: \F_2^{c_W} \rightarrow \F_2^{M}$ that to a syndrome returns an associated error vector, where $c_W$ would be the number of parity checks.
With our notation,  $\tD_W(x)= (E_W)^{-1}(x+\tD'_W(H_W x))$
where (provided the decoder succeeds) $x+\tD'_W(H_W x)\in \ker H_W$
so it is indeed possible to efficiently find a unique preimage of it under the generator matrix $E_W$.}

Then $\tD_W$ is local in the sense that there exists $\eta_W=\Omega(1/\poly\log M)$ such that  for any $e\in\{0,1\}^M$ such that $|e|<\eta_W d_{q,W}$ and any $z\in\{0,1\}^{B_W}$,  then $|\tE_Wz+\tE_W\tD_W(E_W z +e )| \leq (1/\eta_W)  |e|$.
\end{fact}

For later use we define a partial composite decoding operator, which combines decoding results from different stages of the encoding process.

\begin{definition}[Partial composite decoding operator]
\label{def:decode:comp_p}
Let $M, m \in \mathbb{Z}_+$ such that $m$ divides $M$.
Let $W \in \{X, Z\}$.
Let $C_W$ be the linear code with parity-check matrix $H_W$ from Theorem~\ref{thm:qltc}, with codeword length $M$, dimension $B_W$.
Let $\texttt{D}_W$ denote its decoding algorithm.

Let $x = (x_i)_{i \in [M/m]}$ and $w = (w_i)_{i \in [M/m]}$ be tuples of blocks, with each $x_i, w_i \in \mathbb{F}_2^m$.
Define the \emph{partial composite decoding function} as:
 \begin{align*}
\texttt{D}_{p,W,i}(x, w) = \texttt{D}_W(x_1, \ldots, x_i, w_{i+1}, \ldots, w_{M/m}).
\end{align*}
For $i \in \{0, \ldots, M/m\}$, $x \in \mathbb{F}_2^{M}$, $w \in \mathbb{F}_2^{M}$, and $r \in \mathbb{F}_2^{B_W}$, define the partial composite decoding operator as:
 \begin{align*}
D_{p,W,i}(\ket{r} \ket{x} \ket{w}) = \ket{r \oplus \texttt{D}_{p,W,i}(x, w)} \ket{x} \ket{w}.
\end{align*}
Moreover,
for $w \in \mathbb{F}_2^{M}$ and $r \in \mathbb{F}_2^{B_W}$, define the one-sided decoding operator:
 \[
D_{W}(\ket{r} \ket{w}) =
\ket{r \oplus \texttt{D}_{p,W,0}(0, w)} \ket{w}.
\]
For $x \in \mathbb{F}_2^{M}$ and $r \in \mathbb{F}_2^{B_W}$, define the composite decoding operator:

 \begin{align*}
D_{c, W}(\ket{r} \ket{x}) =
\ket{r \oplus \texttt{D}_{p,W,M/m}(x, 0)} \ket{x}.
\end{align*}

\end{definition}

\subsection{PCP of Proximity}

In this section, we recall the notion of probabilistically checkable proofs of proximity (PCPPs), introduced in~\cite{bgh+04}, and their application to the canonical $\cal P$-complete problem \textsc{Circuit Value}. These notions play an important role in our protocols.

\begin{definition}[\cite{bgh+04}]
We define the $\mathcal{P}$-complete language \textsc{Circuit Value} as $\textsc{CktVal} = \{(C, w) : C(w) = 1\}$, where $C$ is a Boolean circuit and $w$ is an input string.
For a pair language $L \subseteq \{0,1\}^* \times \{0,1\}^*$, we define $L(x) = \{y : (x, y) \in L\}$.
\end{definition}

The following definition formalizes a proof system in which the verifier can efficiently check whether a given input is close to being valid, without reading the entire input.

\begin{definition}[PCPs of Proximity (PCPPs), \cite{bgh+04}]
Let $s, \delta: \mathbb{Z}^+ \rightarrow [0, 1]$ be functions. A verifier $V$ is said to be a probabilistically checkable proof of proximity (PCPP) system for a pair language $L$ with proximity parameter $\delta$ and soundness error $s$ if, for every pair of strings $(x, y)$, the following conditions hold:
\begin{itemize}
    \item \textnormal{Completeness:} If $(x, y) \in L$, then there exists a proof string $\pi$ such that $V(x)$ accepts the oracle $y \circ \pi$ with probability $1$.
    \item \textnormal{Soundness:} If $y$ is $\delta(|x|)$-far from $L(x)$, then for every $\pi$, the verifier $V(x)$ accepts the oracle $y \circ \pi$ with probability strictly less than $s(|x|)$.
\end{itemize}
If $s$ and $\delta$ are not explicitly specified, they are assumed to be constants in the interval $(0, 1)$.
\end{definition}

The following theorem, due to~\cite{bgh+04}, establishes that the \textsc{Circuit Value} problem admits a PCPP.

\begin{theorem}[\cite{bgh+04}]
\label{thm:pcpp_}
There exists a constant $\delta_0 \in (0, 1)$ such that for all $n \in \mathbb{Z}^+$ and all $\delta < \delta_0$, the language \textsc{Circuit Value} has a PCP of proximity (for circuits of size $n$) with the following parameters:
\begin{itemize}
    \item randomness complexity $O(\log n)$,
    \item query complexity $O(1)$,
    \item perfect completeness, and
    \item soundness error $1 - \Omega(\delta)$ for proximity parameter $\delta$.
\end{itemize}
\end{theorem}

We now present an amplified version of the above theorem for ease of application:
\begin{lemma}
\label{lem:pcpp_amp}
There exists a constant $\delta_0 \in (0, 1)$ such that for all $n \in \mathbb{Z}^+$ and all $\delta < \delta_0$, the language \textsc{Circuit Value} has a PCP of proximity (for circuits of size $n$) with the following parameters:
\begin{itemize}
\item randomness complexity $O(\log n)$,
\item query complexity $O(\log(s)/\log(1-\Omega(\delta)))$,
\item perfect completeness, and
\item soundness error $s$ for proximity parameter $\delta$.
\end{itemize}
\end{lemma}
\begin{proof}
The lemma follows by repeatedly applying the PCPP verifier from Theorem~\ref{thm:pcpp_} independently at least $N=\log(s)/\log(1-\Omega(\delta))$ times, and accepting if and only if all verifiers accept.

Note that for a no-instance, the acceptance probability satisfies
\begin{align*}
\Pr[\text{verifier accepts}] \leq (1 - \Omega(\delta))^N \leq s.
\end{align*}
\end{proof}

We now define the quantum operator associated with generating a proof in a probabilistically checkable proof of proximity (PCPP).

\begin{definition}[PCPP generating operator]
\label{def:pcpp_generate}
Let $ s,l_x,l_y\in\Z_+$. Let $f : \F_2^{l_x}\times \F_2^{l_y} \to \F_2$ be classical circuits of size $s$. For each $ x\in \F_2^{l_x}, y\in \F_2^{l_y}$, let $f_x(y) = f(x,y)$.
Let $\texttt{N}$ be the classical algorithm that generates the proof for inputs to \textsc{CktVal} from Lemma~\ref{lem:pcpp_amp}. Suppose the output length of $\texttt{N}$ on input of size $ \Theta(s) $ is bounded by $ l_w(s)$.
For $a \in \mathbb{F}_2^{l_w(s)}$, $x \in \mathbb{F}_2^{l_x}$, $y \in \mathbb{F}_2^{l_y}$,  define the operator
\begin{align*}
N_f(\ket{a} \ket{x} \ket{y}) = \ket{a \oplus \texttt{N}(f_x, y)} \ket{x} \ket{y}.
\end{align*}
\end{definition}

\begin{remark}
In our protocol, $x$ typically denotes the output of a circuit $C$ on input $y$.
We define $f_x$ as the circuit $1(C(\cdot) = x )$.
Thus, $(f_x, y) \in \textsc{CktVal}$ if and only if $C(y) = x$.
\end{remark}

Finally, we define the quantum operator corresponding to verifying the PCPP proof.
\begin{definition}[PCPP checking operator]
\label{def:pcpp_check}
Use the same notation as Definition \ref{def:pcpp_generate}.
Let $\texttt{T}$ be the verifier's algorithm from Lemma~\ref{lem:pcpp_amp}.
Let the randomness complexity of $\texttt{T}$ on instance size of $\Theta(n) $ be bounded by $l_r(n)$.

For $a \in \mathbb{F}_2$, $x \in \mathbb{F}_2^{l_x}$, $y \in \mathbb{F}_2^{l_y}$, $w \in \mathbb{F}_2^{l_w(n)}$, and $r \in \mathbb{F}_2^{l_r(n)}$, define
\begin{align*}
T_f( \ket{a} \ket{x} \ket{y} \ket{w} \ket{r} ) = \ket{a \oplus \texttt{T}(f_x, y, w; r)} \ket{x} \ket{y} \ket{w} \ket{r}.
\end{align*}
\end{definition}
\begin{remark}
In our protocol, $r$ will be provided as classical randomness and thus does not consume quantum resources. However, for the sake of analysis, we prefer all quantum states to be represented as pure states.

\end{remark}

\section{Notational Conventions}
\label{sec:notation}

In this section we summarize, for ease of reference, notation that is commonly used across the paper. It is advisable to skip the section at first read but refer to it whenever a notation used in a statement or proof is unclear.

\subsection{General notation}

For $\ell$ an integer we write $[\ell]$ for the set $[\ell]=\{1,\ldots,\ell\}$. For $S$ a finite set, we write $s\in_R S$ for a uniformly random element $s$ of $S$.

\subsection{Codes}
\label{sec:notation-codes}

Our protocols rely on a quantum locally testable code $C$. In addition, we frequently consider the two codes $C_X$ and $C_Z$ that constitute $C$ as a CSS code. We summarize notation used to discuss each code.

\begin{notation}[Codes]\label{not:codes}~

\paragraph{Convention for qLTC code.}
\begin{itemize}
\item $C$ is the quantum LTC code from Theorem~\ref{thm:qltc}, instantiated such that the code has codeword length $M$. Let the code dimension be $B_q$, the distance $d_q$, and the locality $q$. $C$ is locally testable with soundness $\rho$.
\item Let $E_q$ be the encoding unitary for $C$ as defined in Definition~\ref{def:e_qltc}. Let $D_q$ be the decoding unitary as defined in Definition~\ref{def:d_qltc}.
\item Let $ {\cal S}$ be defined as in Lemma \ref{lem:stabilizer-decomposition}.
\end{itemize}

\paragraph{Convention for $X/Z$ codes associated with the qLTC.}
\begin{itemize}
\item $W \in \{X, Z\}$ denotes a single-qubit observable. We often denote by $W'$ the complementary observable, i.e.\ $W'=X$ if $W=Z$ and vice-versa.
\item If $W = X$, let $S_W = H^{\otimes M}$. If $W = Z$, let $S_W = I$.
When there is no ambiguity, we also use $S$ to denote $S_{W}$.
\item $C_W$ is the classical linear code corresponding to the $X$-check matrix $H_X$ when $W = X$, or the $Z$-check matrix $H_Z$ when $W = Z$, as in Theorem~\ref{thm:qltc}. The code has word length $M$, dimension $B_W$. Let $d_{q,X}$ and $d_{q,Z}$ be as defined in Definition~\ref{def:dist_mod}.
\item Let $E_W\in \F_2^{M\times {B_W}}$ denote the generator matrix of the code $C_W$.
\item
Let $\texttt{L}_W$ be the LTC code tester for $C_W$ from Lemma~\ref{lem:ltc_amp}, where the parameter $\varepsilon$ is set to a small constant. The tester has query complexity $Q_W$ and soundness parameter $\kappa_W$.
Define $Q = \max(Q_X, Q_Z)$ and $\kappa = \min(\kappa_X, \kappa_Z)$.
\item Let $\texttt{D}_W$ be the decoding algorithm for $C_W$.
Note that $\tD_W$ is efficient and a local decoder with locality parameter $\eta_W$ in the sense of Fact~\ref{fact:decoder}. Let $\eta_q=\min(\eta_X,\eta_Z)$.
\item
Let $ D_W$ be defined as in Definition \ref{def:decode:comp_p}.
\item Let $ F,G$ be defined as in Definition \ref{lem:logical_qltc}.
Let $K_W=F$ if $W=X$, $K_W=G$ if $W=Z$.
When there is no ambiguity, we also use $K$ to denote $K_{W}$.
\end{itemize}

\end{notation}

\subsection{PCPP}
\label{sec:notation-pcpp}

Regarding probabilistically chechable proofs of proximity, we use the following.

\begin{notation}[PCPP]\label{not:pcpp}~
\begin{itemize}
\item $\texttt{N}$ is the classical algorithm that generates the PCPP proof for inputs to \textsc{CktVal}, and $N_f$ denotes its coherent application as in Definition~\ref{def:pcpp_generate} (on input a classical circuit  $f : \F_2^{l_x}\times \F_2^{l_y} \to \F_2$).
\item $\texttt{T}$ is the PCPP verification algorithm from Lemma~\ref{lem:pcpp_amp} and $T_f$ denotes its coherent application as in Definition~\ref{def:pcpp_check}.
    \item Let $\delta$ and $s$ denote the proximity parameter and soundness error from Lemma~\ref{lem:pcpp_amp}, respectively.
\end{itemize}
\end{notation}

\subsection{Registers}

Since the protocols we design can involve relatively large numbers of quantum registers, we strive to label them consistently using sans serif font as e.g.\ $\sM$, $\sV$, etc.
In the protocols, the following registers will be commonly used.

\begin{notation}[Registers]\label{not:protocol-a}~
\begin{itemize}
    \item $M, m \in \mathbb{Z}_+$ such that $m$ divides $M$. The parameter $m$ represents the size of a local block, and the witness is partitioned into $\rd = M/m$ blocks.

    \item $\mathsf{P}$ is a register of arbitrary dimension, initially held by the prover; it is the prover’s private register.

    \item Let $\mathsf{W} = (\mathsf{W}_{i,j})_{i \in [M/m],\, j \in [m]}$ denote registers initially held by the verifier, containing the witness at the beginning of the protocol. Each $\mathsf{W}_{i,j}$ is a single-qubit register.

    \item Let $\mathsf{X} = (\mathsf{X}_{i,j})_{i \in [M/m],\, j \in [m]}$ and $\mathsf{X}' = (\mathsf{X}'_{i,j})_{i \in [M/m],\, j \in [m]}$, where each $\mathsf{X}_{i,j}$, resp.\ $\mathsf{X}'_{i,j}$, is a single-qubit register. The registers $\mathsf{X}_i$ and $\mathsf{X}'_i$ are used to store copies (in the $W$ basis) of the content in $\mathsf{W}_i$ during the global measurement extraction phase. The prover holds $\mathsf{X}'_i$ and returns $\mathsf{X}_i$ to the verifier.
\end{itemize}

We often write $\mathsf{W}_i$ for the $m$-qubit register $(\mathsf{W}_{i,j})_{j \in [m]}$, and similarly for $\mathsf{X}_i$ and $\mathsf{X}'_i$.
\end{notation}

\subsection{Operators}

The following notation specifies unitary operators that play an important role in the proofs.

\begin{notation}[Operators]\label{not:operators}~
Fix a Boolean function $f:\{0,1\}^{B_q}\to\{0,1\}$. For $W\in\{X,Z\}$ let $W_f = \sum_x (-1)^{f(x)}S_W\proj{x}S_W$.
\begin{itemize}
\item
We label the register on which $D_q$ acts as $\mathsf{W} = (\mathsf{L}, \mathsf{S})$, where $\mathsf{L}$ denotes the register holding the decoded quantum state, and $\mathsf{S}$ denotes the register holding the syndrome.
Let
\begin{align*}
W'_{f, W}=D_{q,\mathsf{W}}^\dagger (W_f)_\mathsf{L} D_{q,\mathsf{W}}.
\end{align*}
When there is no ambiguity, we also use $W'_{f}$ to denote $W'_{f, W}$.
\item
Let $\mathsf{T}$ be a quantum register of $B_W$ qubits.
Let $\tilde{f}:\{0,1\}^{B_W}\to\{0,1\}$ be defined by $\tilde{f}(x)=f( K E_W x)$.
Let
\begin{align*}
W''_{f, W}=
S_{\mathsf{W}}^\dagger
D_{W,\mathsf{TW}}^\dagger (Z_{\tilde{f}})_\mathsf{T} D_{W,\mathsf{TW}} S_{\mathsf{W}}.
\end{align*}
When there is no ambiguity, we also use $W''_{f}$ to denote $W''_{f, W}$.
\item
Let $\mathsf{T}$ be a quantum register of $B_W$ qubits.
Let
 \begin{align}\label{eq:def-qi}
Q_{i, W}=S_{\mathsf{W}}^\dagger D_{p,W,i,\mathsf{TXW}}^\dagger (Z_{\tilde{f}})_\mathsf{T} D_{p,W,i,\mathsf{TXW}}S_{\mathsf{W}}\;.
\end{align}
When there is no ambiguity, we also use $Q_{i}$ to denote $Q_{i, W}$.
\end{itemize}
\end{notation}

\section{Formulation}
\label{sec:qiop}

We proceed with a formal definition of our model, which is a restricted variant of the qIOPs introduced in~\cite{sv-qiop}.

In the prior work, a general framework for quantum interactive oracle proofs (qIOPs) was proposed. Our result also fits within that framework, but our protocol exhibits a more structured and constrained form. We therefore define a restricted model, which we call a \emph{Quantum Interactive Probabilistically Checkable Proof} (qIPCP) system. This name reflects the key characteristics of the model: the verifier begins the protocol holding a quantum witness and verifies it through limited interaction with a prover, making only a small number of queries to the prover’s messages. Hence, the proof is checked in an interactive and probabilistic way.

More specifically: For YES instances, there exists a witness that is initially unentangled with the prover’s private register, such that the verifier accepts with probability at least $2/3$. For NO instances, no matter what witness is used, even if it is initially entangled with the prover’s private register, the verifier accepts with probability at most $1/3$.

Note that this model is equivalent to the one described in Section \ref{sec:results}. In our formulation, for the malicious prover, the witness initially held by the verifier may be an arbitrary quantum state, potentially entangled with the prover’s private register. Thus, this is equivalent to the prover sending a witness to the verifier at the start of the protocol.

We now describe the model more precisely. In each round, the verifier may send the prover a classical message and a subset of the quantum witness register. Importantly, once the verifier sends part of the quantum witness to the prover, he can no longer access it. This is because a malicious prover is not required to return the quantum state at a later stage (even though an honest prover would do so). Furthermore, we explicitly forbid the verifier from returning any of the prover’s previous messages. Therefore, the only quantum messages the verifier can send are parts of the initial quantum witness.

Despite these strong restrictions, we demonstrate that the verifier can still verify QMA-complete languages within this model.

\begin{definition}[Quantum interactive probabilistically checkable proof system]\label{def:qiop}
Let $n \in \mathbb{Z}_+$ denote the instance size, and let
$m : \mathbb{Z}_+ \rightarrow \mathbb{Z}_+$ be a non-decreasing, even-valued function representing the total number of messages as a function of the instance size. Define $h(n) := m(n)/2$.

An $m$-message \emph{quantum interactive probabilistically checkable proof system} (qIPCP) is specified by a tuple of quantum algorithms
$\texttt{V} = (\texttt{V}_1, \texttt{V}_2, \ldots, \texttt{V}_{h+1})$.

For each $i \in [h+1]$, the verifier’s algorithm $\texttt{V}_i$ takes $x$ as input and acts on registers as follows:
\begin{align*}
\big((\mathsf{R}_{i-1}), (\mathsf{R}_1, \ldots, \mathsf{R}_{i-2}, \mathsf{W}_i, \ldots, \mathsf{W}_h)\big)
\longrightarrow
\big((\mathsf{W}_i, \mathsf{J}_i), (\mathsf{R}_1, \ldots, \mathsf{R}_{i-1}, \mathsf{W}_{i+1}, \ldots, \mathsf{W}_h)\big),
\end{align*}
where
\begin{itemize}
\item $\mathsf{R}_i$ is the quantum prover’s message register, for each $i \in \{-1, 0, \ldots, h\}$,
\item $\mathsf{W}_i$ is the quantum witness register, for each $i \in \{1, \ldots, h+1\}$,
\item $\mathsf{J}_i$ is the verifier’s classical message register, for each $i \in \{1, \ldots, h+1\}$.
\item $\mathsf{R}_{-1}$, $\mathsf{R}_0$, and $\mathsf{W}_{h+1}$  are all empty registers.
\item The register $\mathsf{J}$ is initialized to the all-zero classical string.
\item The register $\mathsf{W}=(\mathsf{W}_i)_{i\in[h]}$ is initialized to a quantum witness for the input $x$.
\end{itemize}

Each verifier algorithm $\texttt{V}_i$ must be implementable in polynomial time.

Let $I = \{1, \ldots, h\}$.
A \emph{prover strategy} is defined by a collection of quantum channels $\{\mathbb{P}_i\}_{i \in I}$ (possibly depending on $x$) and a private quantum register $\mathsf{P}$ initialized with an quantum state, where each $\mathbb{P}_i$ acts as:
\[
((\mathsf{W}_i, \mathsf{J}_i), \mathsf{P}) \longrightarrow (\mathsf{R}_i, \mathsf{P}).
\]

The first (classical) bit of the register $\mathsf{J}_{h+1}$ in the output of $\texttt{V}_{h+1}$ represents the verifier’s final decision to accept or reject.
Let $\omega(\texttt{V}, \ket{\phi})$ denote the value of the interactive game defined by $\texttt{V}_x$ when the initial state of the entire system is $\ket{\phi}$. That is, $\omega(\texttt{V}, \ket{\phi})$ is the maximum acceptance probability of $\texttt{V}$ over all prover strategies $\{\mathbb{P}_i\}_{i \in I}$, where $\ket{\phi}$ is defined on the joint quantum register $\mathsf{WP}$.

\paragraph{Resource bounds.}
We define the resource complexity of the system as follows:
\begin{itemize}
    \item \textbf{Communication complexity:} The verifier $\texttt{V}$ has communication complexity $l : \mathbb{Z}_+ \rightarrow \mathbb{Z}_+$ if
    \[
    l(n) \geq \sum_{i \in [h(n)]} l_i(n),
    \]
    where $l_i(n)$ is the total size (in qubits) of the registers $\mathsf{R}_i, \mathsf{J}_i, \mathsf{W}_i$.

    \item \textbf{Query complexity:} The verifier $\texttt{V}$ has query complexity $q : \mathbb{Z}_+ \rightarrow \mathbb{Z}_+$ if
    \[
    q(n) \geq \sum_{i \in [h(n)+1]} q_i(n),
    \]
    where $q_i(n)$ is the number of qubits accessed by $\texttt{V}_i$ in the registers $(\mathsf{R}_1, \ldots, \mathsf{R}_{i-1}, \mathsf{W}_i, \ldots, \mathsf{W}_h)$.
\end{itemize}
\end{definition}

We now formalize the notions of completeness and soundness for a QIPCP system verifying a promise problem.

\begin{definition}[Quantum interactive probabilistically checkable proof systems for a promise problem]\label{def:qiop_}
A quantum interactive probabilistically checkable proof system for a promise problem $(S_{\text{yes}}, S_{\text{no}})$ satisfies the following:

Let $\omega(\texttt{V}, \ket{\phi}_{\mathsf{WP}})$ be defined as in Definition~\ref{def:qiop}. The system must satisfy:

\begin{itemize}
\item \textnormal{Completeness:} There exists a constant $c \in (0,1]$ such that for all $x \in S_{\text{yes}}$, the maximum acceptance probability satisfies ${\max}_{\ket{\psi}}\,\omega(\texttt{V}, \ket{\psi}_{\mathsf{W}}\ket{0}_{\mathsf{P}}) \geq c$. That is, there exists a quantum witness $\ket{\psi}$ and an honest prover strategy under which the verifier accepts with probability at least $c$.

\item \textnormal{Soundness:} There exists a constant $s \in [0, c)$ such that for all $x \in S_{\text{no}}$, the maximum acceptance probability satisfies ${\max}_{\ket{\phi}}\,\omega(\texttt{V}, \ket{\phi}_{\mathsf{WP}}) \leq s$. That is, no combination of prover strategy, prover's initial state, and jointly entangled quantum witness can cause the verifier to accept a no-instance with probability greater than $s$.
\end{itemize}
\end{definition}

We define the main complexity class studied in this paper.

\begin{definition}[Quantum interactive probabilistically checkable proof]
Let $r, l, q: \mathbb{Z}_+ \rightarrow \mathbb{Z}_+$.
We define $\mathcal{QIPCP}(r, l, q)$ as the set of promise problems $(S_{\text{yes}}, S_{\text{no}})$ that have an $r$-message quantum interactive probabilistically checkable proof system with communication complexity $l$ and query complexity $q$.
\end{definition}

\section{QMA-Complete Problems}
\label{sec:qma}

We introduce a variant of the local Hamiltonian problem where each term is diagonal in the $X$ or $Z$ basis respectively, and is furthermore not necessarily local but efficiently sampleable and computable. The $\QMA$-completeness of this variant essentially follows from the results of~\cite{mn25} and~\cite{bmvz25}.

We first introduce the Hamiltonian problem and then state its QMA-completeness. Before that, we define a basic term used in the Hamiltonian.

\begin{definition}[X/Z projection]
For $W\in\{X,Z\}$ we say that a Hamiltonian (hermitian matrix) $H$ is a $W$ projection if and only if $H$ acts non-trivially on a quantum register $\mathsf{S}$ of $k$ qubits as ${I + W_f}/{2}$, where $f: \F_2^{k} \rightarrow \F_2$ is a boolean function.\footnote{Recall the notation $W_f = \sum_x (-1)^{f(x)}S_W\proj{x}S_W$, with $S_W$ a tensor product of $I$ ($W=Z$) or $H$ ($W=X$).} More specifically,
\begin{align*}
H = \Big(\frac{I + W_f}{2}\Big)_{\mathsf{S}}\;.
\end{align*}
\end{definition}

Next, we present the formal formulation of the Hamiltonian and define the corresponding promise problem associated with it.

\begin{definition}[Dual-basis projector Hamiltonian]
\label{def:d_ham}
A Hamiltonian $H \in \C^{2^n} \otimes \C^{2^n}$ is a dual-basis projector Hamiltonian if it is the uniform average of a polynomial number $m=\poly(n)$ of $X$  projection or $Z$  projection terms $H_i$. That is,
\begin{align*}
H = \frac{1}{m} \sum_{i \in [m]} H_i\;.
\end{align*}

Moreover, $H$ is $k$-local if each $H_i$ acts non-trivially on at most $k$ qubits.
\end{definition}

\begin{definition}[Local dual-basis projector Hamiltonian problem]\label{def:dlh}
Let $k \in \Z^+$ and let $a,b:\Z^+ \rightarrow \R^+$ be non-decreasing functions such that $\forall n \in \Z^+,~ a(n) < b(n)$. The local dual-basis projector Hamiltonian problem $\mathcal{DLH}(k,a,b)$ is the promise problem $(S_{yes}, S_{no})$ where
\begin{itemize}
\item $S_{yes}(n)$ is the set of $k$-local dual-basis projector Hamiltonians $H\in \C^{2^n} \otimes \C^{2^n}$ such that $\lambda_{\min}(H) \leq a(n)$,
\item $S_{no}(n)$ is the set of $k$-local dual-basis projector Hamiltonians $H\in \C^{2^n} \otimes \C^{2^n}$ such that $\lambda_{\min}(H) \geq b(n)$.
\end{itemize}
\end{definition}

Finally, we state the QMA-completeness of the Hamiltonian problem. Note that this refers to the local Hamiltonian problem before amplification, where the promise gap is inverse-polynomial.
The following follows immediately from~\cite{mn25}.

\begin{theorem}
\label{thm:dual_ham}
There exist polynomials $p$ and $q$ such that the problem $$\mathcal{DLH}(5, 2^{-p(n)}, 1/q(n))$$ is QMA-complete.
\end{theorem}

\begin{proof}
In~\cite{mn25}, a somewhat incomparable result is shown, where $2^{-p(n)}$ above is replaced by $0$, and the class characterized is $\QMA_1$ (the variant of $\QMA$ with perfect completeness, meaning the verifier always accepts a yes-instance with probability $1$).
The present result follows using the same proof, noting that $\QMA$ is well-known to have verifiers with completeness exponentially close to $1$~\cite{marriott2005quantum}.
\end{proof}

Then, we introduce the Hamiltonian amplification method from~\cite{bmvz25}. We first present their construction, followed by a statement of their result.

Essentially, the idea is to divide the Hamiltonian into commuting layers. In our case, there are two: the $Z$-layer, where each term is diagonal in the computational basis, and the $X$-layer, where each term is diagonal in the Hadamard basis. Each layer is then amplified separately by performing all tests corresponding to that layer simultaneously, and repeating the layer test sequentially $t$ times.

\begin{definition}[Amplified Hamiltonian, \cite{bmvz25}]\label{def:amplified-hamiltonian}
Let $H$ be a $k$-local Hamiltonian on $n$ qubits expressed as an expectation over $m$ projection terms,
\[
H = \frac{1}{m}\sum_{i\in[m]} \Pi_i.
\]
where the local terms of $H$ can be efficiently and equitably partitioned into
$g=O(1)$ commuting layers, indexed by $\chi\in[g]$.
We denote
\[
H_\chi = \mathbb{E}_{i\in[m_\chi]}[\Pi_i^{\chi}],
\quad
w_\chi = \frac{m_\chi}{m} = \Omega(g^{-1}),
\]
where $w_\chi$ represents the relative weight of clauses assigned to layer~$\chi$, $ H = \E_{\chi} H_\chi$.

For any integer parameter $t\ge 1$, we define the \emph{amplified Hamiltonian} as
\begin{align}
H^{(t)}
=
I - \mathbb{E}_{\chi}
\bigotimes_{j\in [t]}\!\Big(
\prod_{i\in[m_\chi]} (I - \Pi_i^{\chi})
\Big)
,
\end{align}
where, $\mathbb{E}_{\chi}$ indicates sampling $\chi$ from the distribution $(w_1, w_2, \ldots, w_g)$ over $[g]$.
\end{definition}

By applying this procedure, the minimum energy of the no-instance can be amplified until it exceeds a certain threshold.

\begin{theorem}[Amplification from the DL, \cite{bmvz25}]\label{thm:DL-amplification}
Assume that the minimum eigenvalue of $H$ satisfies $\lambda_{\min}(H)\ge \gamma$.
Then, the minimum eigenvalue of the amplified Hamiltonian $H^{(t)}$ satisfies
\[
\lambda_{\min}\big(H^{(t)}\big)
\ge
\min\Big\{
\Theta(g^{-2}),\;
\Omega\!\Big(\frac{t\cdot m \cdot \gamma}{g^{4}}\Big)
\Big\}.
\]
\end{theorem}

Finally, we apply the gap amplification theorem from~\cite{bmvz25} to the dual-basis projector Hamiltonian problem.

\begin{corollary}
\label{cor:dual_ham_amp}
There exists a polynomial $p$ and a constant $c_{\LH}>0$ such that the problem $$\mathcal{DLH}(\poly(n), 2^{-p(n)}, c_{\LH})$$ is QMA-complete.
\end{corollary}
\begin{proof}
The completeness follows from the fact that, within each layer, all projectors commute. The maximum eigenvalue $\lambda_{\max}(\prod_{i\in[m_\chi]} (I - \Pi_i^{\chi}))$ is lower bounded by
$1 - \min_{\ket{\phi}} \sum_i \bra{\phi}\Pi_i^{\chi}\ket{\phi}$.
In the yes-instance, there exists a state $\ket{\phi}$ such that each $ \bra{\phi}\Pi_i^{\chi}\ket{\phi} \le 1/\exp(n)$, since there are at most $\mathrm{poly}(n)$ such terms.
Moreover, because each layer contains at most $\mathrm{poly}(n)$ terms, the completeness follows.

The soundness follows from Theorem \ref{thm:dual_ham} and Theorem \ref{thm:DL-amplification}.

\end{proof}

\section{Evaluating Observables on Corrupted Codewords of a CSS Code}
\label{sec:css}

This section establishes some utility lemma on the structure of corrupted codewords of a quantum CSS code. The main lemma of the section is Lemma~\ref{lem:DO2Dq_beta}, which compares two different methods of evaluating a ``logical'' operation on a corrupted codeword.

We first formalize the notion of local testing for a quantum locally testable code (qLTC).

\begin{mdframed}
\begin{definition}[Local testing of a qLTC code]
\label{def:qltc-local-testing}
Let $C=CSS(H_X, H_Z)$ be a stabilizer code on $n$ qubits, and let $\ket{\phi}$ be a quantum state. The verifier performs the following steps $c$ times:
\begin{itemize}
\item The verifier randomly samples $W \in \{X, Z\}$.
\item The verifier randomly samples a row $r_W$ from $H_W$.
\item The verifier measures the observable $W(r_W)$ on the state $\ket{\phi}$ and rejects if the measurement outcome is $-1$.
\end{itemize}
If none of the tests reject, the verifier accepts.
\end{definition}
\end{mdframed}

\begin{lemma}
\label{lem:qltclocaltest}
Let $C=CSS(H_X, H_Z)$ be a stabilizer code on $n$ qubits such that $H_X\in \F_2^{c_X\times M}$, $H_Z\in \F_2^{c_Z\times M}$, are parity check matrices associated with classical locally testable codes with soundness $\rho$. Let $ {\cal S}$ be defined as in Lemma \ref{lem:stabilizer-decomposition}. Let
\[
\ket{\psi} = \sum_{(a,b)\in {\cal S}} \alpha_{a,b} X(a)Z(b) \ket{\eta_{a,b}}
\]
be a quantum state, where $\{\ket{\eta_{a,b}}\}$ are a family of codewords.
Then the verifier in Definition~\ref{def:qltc-local-testing} accepts $\ket{\psi}$ with probability at most
\[
\Pr[\text{verifier accepts}]
\leq \sum_{(a,b)\in {\cal S}} |\alpha_{a,b}|^2\cdot \Big(1  - \frac{ \rho }{2M}\cdot wt_{\text{cent}}(X(a)Z(b)) \Big)^c\;.
\]
\end{lemma}
\begin{remark}
A similar statement was first proposed in~\cite{ae15}. The result was further strengthened in~\cite{eh17}, although no proof was provided. We strengthen the result even further.
\end{remark}

\begin{proof}
Let the verifier in Definition~\ref{def:qltc-local-testing} be denoted $V$.
Suppose the sequence of observables measured by $V$ is $(W_i(r_{i}) )_{i\in[c]}$. Note that these observables pairwise commute. Thus
\begin{align*}
\Pr\big[\text{$V$ accepts } \mid (W_i(r_{i}) )_{i\in[c]}\big]
&=\Big\| \prod_{i\in [c]}  \frac{I+ W_i(r_{i})}{2} \ket{\phi}\Big\|^2 \\
&= \bra{\phi} \prod_{i\in [c]}  \frac{I+ W_i(r_{i})}{2} \ket{\phi} \\
=&\sum_{(a,b)\in {\cal S}} |\alpha_{a,b}|^2 \bra{\eta_{a,b}}Z(b) X(a) \Big(\prod_{i\in [c]}  \frac{I+ W_i(r_{i})}{2} \Big) X(a)Z(b) \ket{\eta_{a,b}}\;.
\end{align*}
Here, the second step follows from $(\frac{I+ W_i(r_{i})}{2})^2 = \frac{I+ W_i(r_{i})}{2}$, and the third step follows from the fact that the $X(a)Z(b) \ket{\eta_{a,b}}$ are orthonormal eigenvectors of the $W_i(r_{i})$.
Continuing, since the $\ket{\eta_{a,b}}$ are codewords,
\begin{align*}
 \bra{\eta_{a,b}}Z(b) X(a) \Big(\prod_{i\in [c]}  \frac{I+ W_i(r_{i})}{2} \Big) X(a)Z(b) \ket{\eta_{a,b}}
=
\prod_{i\in [c]} 1_{[X(a)Z(b), W_i(r_{i})]=0}\; .
\end{align*}
Averaging over the verifier's choices of $W$ and $r_i$,
\begin{align}
\Pr[\text{$V$ accepts}]
=&\sum_{(a,b)\in {\cal S}} |\alpha_{a,b}|^2 \cdot
\Big( \frac{1}{2} \cdot\E_{i\in_R[c_X]} 1_{[X(a)Z(b), X((H_X)_i)]=0} +\frac{1}{2} \cdot \E_{i\in_R[c_Z]}
1_{[X(a)Z(b), Z((H_Z)_i)]=0}
\Big)^c \notag\\
=& \sum_{(a,b)\in {\cal S}} |\alpha_{a,b}|^2\cdot \Big( \frac{c_X - wt(H_X b)}{2 c_X}+ \frac{c_Z - wt(H_Z a)}{2 c_Z}  \Big)^c\; .\label{eq:ltc-1}
\end{align}
Let $C_X$ be the locally testable code corresponding to the
X-check matrix $H_X$.
Let $C_Z$ be the locally testable code corresponding to the
Z-check matrix $H_Z$.
By Definition \ref{def:ltc},
\begin{align*}
wt(H_Z a) \geq&~ \rho \cdot \frac{c_Z}{M} \cdot d_H(a, C_Z)\;, \\
wt(H_X b) \geq&~ \rho \cdot
 \frac{c_X}{M}\cdot  d_H(b, C_X)\; .
\end{align*}
Moreover, for any $ W\in \{X, Z\}$ and $s\in\F_2^n$,
let $ W'$ be the complementary basis of $W$, then
\begin{align*}
d_H(s, C_W) =&~ \min_{r \in \F_2^M,  H_W \cdot r=0} \big\{ {wt}(r + s) \big\} \\
\geq&~ \min_{a', b' \in \F_2^n,  H_X b'=0, H_Z a' = 0} \big\{ {wt}(X(a')Z(b') (W')(s)) \big\}\\
=&~ wt_{\text{cent}}(W'(s))\;,
\end{align*}
and
\begin{align*}
wt_{\text{cent}}(X(a)) + wt_{\text{cent}}(Z(b)) \geq wt_{\text{cent}}(X(a)Z(b))\;.
\end{align*}
Plugging back into~\eqref{eq:ltc-1},
\begin{align*}
\Pr[\text{$V$ accepts}]
\leq&~ \sum_{(a,b)\in {\cal S}} |\alpha_{a,b}|^2\cdot \Big(1  - \frac{ \rho }{2M}\cdot d_H(b, C_X)  -  \frac{\rho }{2M} d_H(a, C_Z) \Big)^c \\
\leq&~ \sum_{(a,b)\in {\cal S}} |\alpha_{a,b}|^2\cdot \Big(1  - \frac{ \rho }{2M}\cdot wt_{\text{cent}}(X(a)Z(b)) \Big)^c\; .
\end{align*}
\end{proof}

The next two lemmas, albeit simple, are used repeatedly throughout the proofs. In particular, they appear in the proof of the lemma that follows them.

\begin{lemma}
\label{lem:q_noise_s}
Recall the notation introduced in Section~\ref{sec:notation}.
For any $u, v \in C_W$, if $|u - v| < d_{q, W}$, then
\[
Ku = Kv.
\]
\end{lemma}

\begin{proof}
Let $ W'$ be the complementary basis of $W$.

Since $u, v \in C_W$, it follows that $u - v \in C_W$.
Moreover,
\[
|u - v| < d_{q, W} =
\min\{\, |w| : w \in \ker(H_W) \setminus \Im(H_{W'}^T) \,\},
\]
implies that $u - v \in \Im(H_{W'}^T)$.

Because $K H_{W'}^T = 0$ by~\eqref{eq:kfg},
we conclude that $K(u - v) = 0$,
and hence $Ku = Kv$.
\end{proof}

\begin{lemma}
\label{lem:q_noise}
Recall the notation introduced in Section~\ref{sec:notation}.
Then, for any $p \in \F_2^M$ such that $|p| < \eta_W d_{q,W}$, then for any $v\in \F_2^{B_W} $, it holds that
\[
K E_W\, \texttt{D}_W(p + E_W v) = K E_W v.
\]
\end{lemma}

\begin{proof}
By Fact~\ref{fact:decoder},
\[
\big|\, \tE_W v + \tE_W \tD_W(p + E_W v)\, \big| \leq \frac{1}{\eta_W} \cdot |p| < d_{q,W}.
\]

Then, by Lemma~\ref{lem:q_noise_s},
\[
K E_W\, \texttt{D}_W(p + E_W v) = K E_W v. \qedhere
\]
\end{proof}

The following important lemma relates measurements obtained on the decoded quantum word (the observable $W'$), to measurements obtained directly on the encoded word but subsequently decoded according to the classical code (the observable $W''$). The lemma will be a key workhorse of soundness arguments in the following sections.

\begin{lemma}
\label{lem:DO2Dq_beta}
Recall the notations introduced in Section~\ref{sec:notation}.
Let
\begin{align*}
\ket{\beta} = \sum_{(a,b)\in {\cal S}} 1(wt_{\text{cent}}(X(a)Z(b)) < \eta_q d_q )\cdot \alpha_{a,b} X(a)Z(b) \ket{\eta_{a,b}}
\end{align*}
be a quantum state, where $\{\ket{\eta_{a,b}}\}$ are codewords.
Then, for any Boolean function $f : \{0, 1\}^{B_q} \rightarrow \{0, 1\}$,
\begin{align*}
W''_{\Tilde{f}, W}\ket{\beta } =  W'_{f, W} \ket{\beta}.
\end{align*}

\end{lemma}
\begin{remark}

Note that, for $v\in \F_2^{B_W}$, $ \texttt{D}_W(E_W v) =v$.
\end{remark}

\begin{proof}
{\bf Part I: analysis of the behavior of $W'_f$.}
Let $\ket{\theta_{a,b}}$ be such that $\ket{\eta_{a,b}} =E_q(\ket{\theta_{a,b}})$.
For $x\in\mathbb{F}_2^{B_q}$ let $ \ket{\phi_x}= E_q(\ket{x})$ be as in Definition \ref{def:e_qltc}.
Let $ \ket{\psi_x} = \sum_{i\in \F_2^{B_q}} (-1)^{i\cdot x} \ket{\phi_i}$
Let $f'_W : \{0, 1\}^{M }\rightarrow \{0, 1\}$ be $f'_W(x) = f( K_W x)$.

Note that in the definition of $\ket{\phi_x}$ from Definition~\ref{def:e_qltc}, the projection
\[
\Pi_x = \prod_{j \in [k]} \frac{I + (-1)^{x_j} Z(G_j)}{2}
\]
is diagonal in the computational basis. For any $y \in \F_2^M$ in the support of $\Pi_x$, we have $G_j y = x_j$. Therefore, $f'(y) = f(Gy) = f(x)$, and
\begin{align*}
&~ Z_{f'} E_q(\ket{x}) \\
=&~ (-1)^{f(x)} \cdot E_q(\ket{x}),
\end{align*}
and similarly,
\begin{align*}
&~ X_{f'} E_q(H\ket{x}) \\
=&~ (-1)^{f(x)} \cdot E_q(H\ket{x}).
\end{align*}

Using the definition of $F$ and $G$,
\begin{align*}
Z_{f'}\ket{\phi_x} = (-1)^{f(x)}\cdot E_q(\ket{x})=E_q(Z_f \ket{x})\;,
\end{align*}
and
\begin{align*}
X_{f'}\ket{\psi_x} = (-1)^{f(x)}\cdot E_q(H\ket{x})=E_q(X_f H\ket{x})\;,
\end{align*}
Expanding $\ket{\eta_{a,b}}$ in the codeword basis given by the $ \ket{\phi_x}$ or $ \ket{\psi_x}$ it follows by linearity that
\begin{align*}
 W_{f'} \ket{\eta_{a,b}} = E_q(W_f \ket{\theta_{a,b}})\;.
\end{align*}

By the definitions of $d_q$, $W'_f$, and $D_q$ in Notation \ref{not:codes}, we have, 
for $(a,b)\in {\cal S}$,
\begin{align}
W'_f X(a)Z(b)\ket{\eta_{a,b}}
&=
D_{q}^\dagger W_f D_{q}
X(a)Z(b)\ket{\eta_{a,b}} \notag\\
&=
D_{q}^\dagger W_f
\ket{\theta_{a,b}}\ket{H_Z a} \ket{H_X b} \notag\\
&=
D_{q}^\dagger
D_{q}(W_{f'} \ket{\eta_{a,b}})
\ket{H_Z a} \ket{H_X b} \notag\\
&= X(a)Z(b) W_{f'}\ket{\eta_{a,b}} . \label{eq:ow-1}
\end{align}

{\bf Part II: analysis of the behavior of $W''_{w^T K E_W }$.}
Let $ v\in \F_2^B$,
\begin{align*}
W''_{w^T K E_W } S \ket{E_W v} = (-1)^{w^T K E_W v}S\ket{E_W v} = W(w^T K)S\ket{E_W v} .
\end{align*}

Therefore,
\begin{align*}
W''_{w^T K E_W } \ket{\eta_{a,b}} = W(w^T K)\ket{\eta_{a,b}}.
\end{align*}

Moreover, we have that
\begin{align*}
&~D_{Z,\mathsf{TW}}^\dagger Z_{\Tilde{f}, \mathsf{T} }D_{Z,\mathsf{TW}} X(p)Z(q) \ket{E_Z v} \\
=&~ (-1)^{p^T q} D_{Z,\mathsf{TW}}^\dagger Z_{\Tilde{f}, \mathsf{T} }D_{Z,\mathsf{TW}} Z(q) X(p) \ket{E_Z v} \\
=&~ (-1)^{p^T q} D_{Z,\mathsf{TW}}^\dagger Z_{\Tilde{f}, \mathsf{T} }D_{Z,\mathsf{TW}} Z(q)  \ket{p+E_Z v} \\
=&~ (-1)^{p^T q} Z(q) D_{Z,\mathsf{TW}}^\dagger Z_{\Tilde{f}, \mathsf{T} }D_{Z,\mathsf{TW}}  \ket{p+E_Z v} \\
=&~ (-1)^{p^T q + f(KE_Z \texttt{D}_Z(p+E_Z v)) } Z(q) \ket{p+E_Z v} \\
=&~ (-1)^{p^T q + f(KE_Z \texttt{D}_Z(p+E_Z v)) } Z(q)X(p) \ket{E_Z v} \\
=&~ (-1)^{ f(KE_Z \texttt{D}_Z(p+E_Z v)) } X(p)Z(q) \ket{E_Z v} .
\end{align*}

Similarly, we have that
\begin{align*}
D_{X,\mathsf{TW}}^\dagger X_{\Tilde{f}, \mathsf{T} }D_{X,\mathsf{TW}} X(p)Z(q) H\ket{E_X v}  = (-1)^{ f(KE_X \texttt{D}_X(q+E_X v)) } X(p)Z(q) H\ket{E_X v} .
\end{align*}

Moreover,
\begin{align*}
&~X(p)Z(q) Z_{f'} \ket{E_Z v} \\
=&~(-1)^{f(K E_Z v)} X(p)Z(q) \ket{E_Z v} ,
\end{align*}
and
\begin{align*}
&~X(p)Z(q) X_{f'} H \ket{E_X v} \\
=&~(-1)^{f(K E_X v)} X(p)Z(q) H \ket{E_X v} .
\end{align*}

Since $ wt_{\text{cent}}(X(a)Z(b)) < \eta_q d_q$,
we have that $ |a|, |b| < \eta_q d_q$. By Lemma \ref{lem:q_noise},
\begin{align*}
G E_Z \texttt{D}_Z(a+E_Z v) = G E_Z v\;,\qquad F E_X \texttt{D}_X(b+E_Xv) = F E_X v\;.
\end{align*}
Therefore,
\begin{equation}\label{eq:ow-2}
W''_{\Tilde{f} }X(a)Z(b) \ket{\eta_{a,b}} = X(a)Z(b) W_{f'}\ket{\eta_{a,b}}.
\end{equation}

{\bf Part III: Putting things together.}
Combining~\eqref{eq:ow-1} and~\eqref{eq:ow-2},
\begin{align*}
W'_{f} X(a)Z(b)\ket{\eta_{a,b}} = W''_{ \Tilde{f} }X(a)Z(b) \ket{\eta_{a,b}},
\end{align*}
and
\begin{align*}
W'_f \ket{\beta} = W''_{ \Tilde{f} }\ket{\beta}.
\end{align*}

\end{proof}

\begin{lemma}
\label{lem:DO2Dq}
Recall the notations introduced in Section~\ref{sec:notation}.
Let
\[
\ket{\psi} = \sum_{(a,b)\in {\cal S}} \alpha_{a,b} X(a)Z(b) \ket{\eta_{a,b}}
\]
be a quantum state, where $\{\ket{\eta_{a,b}}\}$ are a family of codewords. Then, for any boolean function $f :\mathbb{F}_2^{B_q}\rightarrow \mathbb{F}_2$,
\begin{align*}
\|W''_{\Tilde{f} }\ket{\psi } -  W'_f \ket{\psi}\|^2 \leq \Theta(1)\cdot \sum_{(a,b)\in {\cal S}} 1(wt_{\text{cent}}(X(a)Z(b)) \geq \eta_q d_q)\cdot |\alpha_{a,b}|^2\;.
\end{align*}
\end{lemma}

\begin{proof}
This follows directly from Lemma \ref{lem:DO2Dq_beta}.
\end{proof}

\section{PCP of Proximity Applied to Quantum States}
\label{sec:pcp-proxim}

In this section we introduce and analyze a simple interaction between prover and verifier in which the prover is asked to compute a PCPP for the validity of a certain computation (abstracted through the function $f$ below) when given an encoded input to the computation and a target output. This will be used as a building block later on.

\subsection{Actions of the honest prover}
In the following definition the prover is given registers $ \mathsf{U}$, $\mathsf{B}$,  and $\mathsf{E}$.
The register $\mathsf{U}$ represents the (target) output of a computation. The register $\mathsf{B}$ represents the witness for the verification of the computation, which the prover is to compute: initially, the register $ \mathsf{B}$ is all $\ket{0}$.
Finally, $\mathsf{E}$ is prepared as the encoded, according to a classical error correcting code, input of the computation.

Note that, to maintain the flexibility and generality of the results in this section, we do not adopt a specific classical code (e.g., the one defined in Notation~\ref{not:codes}); instead, we state the results for any classical error-correcting code $C_P$ with appropriate parameters.

\begin{remark}
In this section as in subsequent ones we describe the actions of a honest prover in a coherent form, i.e.\ replacing measurements by copying information, etc. This is to preserve the protocol’s flexibility, e.g., the ability to generalize it to handle more complex observables. Any prover can be written in such a purified form, even though a dishonest prover may of course apply any unitaries they like.
\end{remark}

\begin{mdframed}

\begin{definition}[PCP of proximity for verifying a computation, honest prover]
\label{def:pcpp:hp}
Recall the notations introduced in Section~\ref{sec:notation}.

Let $C_P$ be a classical error-correcting code with dimension $k$, codeword length $n_P(k)$, decoder $\texttt{D}_P$, and code tester $\texttt{T}_P$.

Let
$\mathsf{U}$ be an $l_x$-qubit register denoting the output, $\mathsf{E}$ an $l_y$-qubit register denoting the encoded input, let $f:\F_2^{n_P^{-1}(l_y)} \rightarrow \F_2^{l_x}$ represent the verification of the computation.
Let $g:\F_2^{l_x}\times \F_2^{l_y} \rightarrow \F_2$ be defined as
\begin{align*}
g(x,y) = 1(f(\texttt{D}_P(y))=x) \wedge \texttt{T}_P(y)\;.
\end{align*}
Let $s$ be the size of the circuit $ g $.
Let $\mathsf{B}$ be a quantum register of  $l_{w}(s) $ qubits.
The honest prover performs the following steps:
\begin{itemize}
\item The prover computes an $l_{w}(s)$-bit long PCPP proof for a statement regarding the content of register $\mathsf{E}$. This proof is stored in register $\mathsf{B}$.
Formally,
the prover applies the following unitary:
\begin{align*}
N_{g, \mathsf{B} \mathsf{U} \mathsf{E}}\;.
\end{align*}

\item The prover sends the verifier the quantum registers $\mathsf{B}$.
\end{itemize}
\end{definition}

\end{mdframed}

\subsection{Actions of the verifier}

Before running the following protocol, the honest prover has already sent $ \mathsf{U}$, the register storing the computation output, to the verifier.
\begin{mdframed}
\begin{definition}[PCP of proximity for verifying a computation, verifier]
\label{def:pcpp:v}
Use
the notation in Definition \ref{def:pcpp:hp}.
The verifier performs the following steps:
\begin{itemize}
\item The verifier executes the PCPP verifier on input $(g_x,y)$ and stores the verification result in a single qubit register $\mathsf{C}$, where $x$ is the contents of register $\mathsf{U}$, $y$ is the contents of register $\mathsf{E}$ and the proof is in register $\mathsf{B}$. The randomness is stored in register $\mathsf{R}$.
Formally, the verifier initialize $\mathsf{C}$ to $\ket{0}$ and $\mathsf{R}$ to $l_r(s)$  classical randomness bits, then the verifier applies $$T_{g, \mathsf{C}\mathsf{U}\mathsf{E}\mathsf{B}\mathsf{R}}.$$
\item The verifier rejects if the PCPP verifier rejects. Formally, the verifier  measures $\mathsf{C}$ and rejects if the outcome is $0$.
\end{itemize}

\end{definition}
\end{mdframed}

\subsection{Main results}

We now state the main lemma in this section. Informally, the lemma states that as long as every input within the proximity parameter of one that decodes to the target output also decodes to that output, the verifier’s acceptance probability is at most the sum of the PCPP soundness parameter and the probability that a measurement of the whole quantum system yields a classical string that is a good instance. Note that we can obtain a PCPP verifier with improved soundness (while preserving the proximity parameter) by repeatedly applying the PCPP verifier, at the cost of increased query complexity.

\begin{lemma}[Implication of the PCPP test]
\label{lem:pcpp}
Let $ s,\delta$ be defined as in Notation \ref{not:pcpp}. Use the notation introduced in Definition \ref{def:pcpp:v}.

Let $\mathsf{O}$ be a quantum register representing all registers in the system except for $\mathsf{U}$, $\mathsf{E}$, and $\mathsf{B}$.
Let $\ket{\phi}_{\mathsf{O}\mathsf{U}\mathsf{E}\mathsf{B}}$ be the quantum state of the system before the prover returns the quantum registers $\mathsf{B}$.
Let the state be written as
\begin{align*}
\ket{\phi} = \sum_{o,u,e,b} \alpha_{o,u,e,b} \ket{o}_{\mathsf{O}}\ket{u}_{\mathsf{U}}\ket{e}_{\mathsf{E}}\ket{b}_{\mathsf{B}}.
\end{align*}

For $u\in\F_2^{l_x}$, let $S_u=\{e\in\F_2^{l_y}\mid 1(f(\texttt{D}_P(e))=u) \wedge \texttt{T}_P(e)=1\}$. If $f(\texttt{D}_P(e))=u$ for every $u\in\F_2^{l_x}$ and every $e\in\F_2^{l_y}$ with $\text{dist}(e,S_u)\le\delta$,
then
\begin{itemize}
\item The acceptance probability of the verifier is upper bounded by
\begin{align*}
\Pr[\text{PCPP test accepts}] \leq  \max\Big\{s , \sum |\alpha_{o,u,e,b}|^2 \cdot 1\Big( f(\texttt{D}_P(e))=u \Big)  \Big\} .
\end{align*}
\item The rejection probability of the verifier is lower bounded by
\begin{align*}
\Pr[\text{PCPP test rejects}] \geq  - s + \sum |\alpha_{o,u,e,b}|^2 \cdot 1\Big(f(\texttt{D}_P(e)) \neq u \Big) .
\end{align*}
\end{itemize}
\end{lemma}
\begin{proof}

For $ u\in \F_2^{l_x}$, recall
\begin{align*}
S_{u} = \{ e \in \F_2^{l_y} \mid
1(f(\texttt{D}_P(e))=u) \wedge \texttt{T}_P(e)=1 \}.
\end{align*}

We analyze the acceptance probability:
\begin{align*}
\Pr[\text{verifier accepts}]
= &~ \sum |\alpha_{o,u,e,b}|^2 \cdot \Pr[\text{verifier accepts } \ket{o,u,e,b}].
\end{align*}

\begin{itemize}
\item If $\text{dist}(e, S_{u}) > \delta$, then
  \begin{align*}
\Pr[\text{PCPP accepts } u, e, b] \leq &~ s.
\end{align*}
\item If $\text{dist}(e, S_{u}) \leq \delta$, then
\begin{align*}
\Pr[\text{PCPP accepts } u, e, b] \leq &~ 1= 1\Big(f(\texttt{D}_P(e))=u \Big),
  \end{align*}
where the last step uses the condition of the lemma, that $\text{dist}(e,S_u)\le\delta$ implies $f(\texttt{D}_P(e))=u$.
\end{itemize}
As a result,
\begin{align*}
\Pr[\text{verifier accepts}]
\leq  &~ \sum |\alpha_{o,u,e,b}|^2 \cdot \max\Big\{s , 1\Big( f(\texttt{D}_P(e))=u  \Big) \Big\} \\
\leq  &~ \max\Big\{s , \sum |\alpha_{o,u,e,b}|^2 \cdot 1\Big( f(\texttt{D}_P(e))=u \Big)  \Big\} \;,
\end{align*}
and
\begin{align*}
\Pr[\text{verifier rejects}] \geq  - s + \sum |\alpha_{o,u,e,b}|^2 \cdot 1\Big( f(\texttt{D}_P(e))\neq u  \Big)  \;.
\end{align*}

\end{proof}

\section{Global Measurement Extraction}
\label{sec:gme}

In our protocol, the quantum proof is given as a qLTC encoding of a witness. The verifier first checks the correctness of the qLTC encoding. Ultimately, the witness used will be an amplified version of a witness for the QMA language instance provided as input. In the following, we present the general procedure for verification of the encoding, when applied to an arbitrary witness.

In the end, our goal is to measure observables of the form $X_f$ or $Z_f$ on the physical qubits. Since the procedures are very similar, we use the case of measuring $Z_f$ as an example. To perform this measurement, the verifier must extract measurement outcomes in the Z basis. This overall process is referred to as \emph{Global Measurement Extraction}.
The Global Measurement Extraction procedure consists of repeatedly applying the \emph{Single-Round Measurement Extraction} subroutine to successive portions of the qLTC-encoded state until all qubits are measured.

The Single-Round Measurement Extraction process works as follows: the verifier sends a small portion of the encoded witness, consisting of $m$ qubits, to the prover. The only requirement is that this portion must be much smaller than the recoverable distance of the code, i.e.\ an error of  weight at most $m$ can be successfully corrected by the code decoder.

For concreteness, consider the case where the verifier sends the first $m$ qubits.
The prover is then asked to measure these $m$ qubits in the Z basis, and to reply with
 the measurement results themselves.
The prover also retains a copy of the measurement outcomes.

Later in the protocol, the reported measurement outcomes on this portion of the qubits are used directly. Whenever the verifier wishes to access a measured qubit, he reads the corresponding reported bit; this is local, as each access touches only a single bit. The reported outcomes are kept consistent with an encoded witness through the local testability of the quantum code's classical part.

After receiving the prover’s message, the verifier checks the Z part of the qLTC code. Specifically, the verifier randomly samples a Z-type stabilizer. Since stabilizers are local, they act on only a small number of qubits. For qubits that the verifier holds, he can measure them directly. For qubits that the prover already measured in the Z-basis, the verifier reads the reported bits directly.

Note that the prover could introduce additional errors into the reported outcomes during the Single-Round Measurement Extraction. Because Z (phase) errors do not affect Z basis measurement outcomes, we only need to worry about X (bit-flip) errors. The Z-LTC test performed after receiving the prover’s message ensures that all X errors remain bounded to a small level.

The verifier then applies the Single-Round Measurement Extraction process to the next portion of the remaining qLTC-encoded witness. This process is repeated until all physical qubits have been measured in the Z-basis.

\subsection{Actions of the honest prover}

We now describe formally the basic procedure through which the prover extracts measurement results from a small portion of the quantum state.
This process, which we refer to as \emph{Single Round Measurement Extraction}, enables the prover to verifiably encode measurement outcomes for a small block of qubits. Even though in the protocol itself the prover is required to send back classical information, for purposes of modeling and analysis we will consider the actions of a purified prover, i.e.\ such that measurement outcomes are stored coherently in an outcome register. This register is later to be measured, as it is sent from prover to verifier.

Before we start, we introduce some notation that will be useful for presenting the prover’s protocol.

\begin{definition}[Controlled unitary]
Let $\mathsf{P}$ be a register and $U_\mathsf{P}$ a unitary acting on it. Let $\mathsf{R}$ be a one-qubit register.
The controlled unitary $C(U_\mathsf{P}, \mathsf{R})$ applies $U_\mathsf{P}$ if the qubit in $\mathsf{R}$ is $|1\rangle$:
\begin{align*}
    C(U_\mathsf{P}, \mathsf{R}) = \begin{pmatrix} I_\mathsf{P} & 0 \\ 0 & U_\mathsf{P} \end{pmatrix}\;.
\end{align*}
\end{definition}

\begin{definition}[Purified measurement]
Let $\mathsf{P}, \mathsf{R}$ be registers and $O_\mathsf{P}$ an observable decomposed as $O = \sum_{i=1}^r a_i P_i$, where each $P_i$ is an orthogonal projection and $I = \sum_i P_i$. The unitary that performs the purified measurement and stores the outcome in $\mathsf{R}$ is:
\begin{align*}
    P(O_\mathsf{P}, \mathsf{R}) = \sum_{i=1}^r P_{i, \mathsf{P}} \cdot (X(a_i))_\mathsf{R}\;.
\end{align*}
\end{definition}

The honest prover's actions in this procedure are defined below.

\begin{mdframed}

\begin{definition}[Single round measurement extraction, honest prover]
\label{def:lme:hp}
Let $W$ be a single-qubit observable.
The honest prover performs the following steps:
\begin{itemize}
\item The verifier sends the prover a quantum register $\mathsf{W}$ consisting of $m$ qubits.
\item The prover measures the observable $W$ $m$ times, on each of the qubits of $\mathsf{W}$, and stores two copies of each result, in single-qubit registers $\mathsf{X}_i$ and $\mathsf{X}'_i$ respectively.

Formally, the prover applies the following unitary:
\begin{align*}
\prod_{i\in[m]} P(W_{\mathsf{W}}(e_i), \mathsf{X}_i) P(W_{\mathsf{W}}(e_i), \mathsf{X}'_i)\;.
\end{align*}

\item Finally, the prover sends the verifier the quantum registers $\mathsf{WX}$.
\end{itemize}
\end{definition}

\end{mdframed}

We extend the single round measurement extraction procedure to a global measurement extraction procedure.
The global measurement extraction process involves repeating the single round measurement extraction across multiple disjoint portions of the qLTC encoded state.
We formally define the honest prover's actions for this procedure below.
\begin{mdframed}
\begin{definition}[Global measurement extraction, honest prover]
\label{def:gme:hp}
Recall the notations introduced in Section~\ref{sec:notation}.
The honest prover performs the following steps, for $i=1,\ldots,M/m$:
\begin{itemize}
    \item The prover executes the actions described in Definition~\ref{def:lme:hp}, using $W=W, \mathsf{X} = \mathsf{X}_i, \mathsf{X}' = \mathsf{X}'_i$.
\end{itemize}
\end{definition}
\end{mdframed}

\subsection{Actions of the verifier}

We now describe the verifier's actions in the single round measurement extraction process, corresponding to the honest prover's strategy previously defined.

\begin{mdframed}
\begin{definition}[Single round measurement extraction, verifier]
\label{def:lme:v}
Let $W$ be a single-qubit observable.
The verifier performs the following steps:
\begin{itemize}
\item The verifier sends the prover a quantum register $\mathsf{W}$ consisting of $m$ qubits.
\item The prover sends back the quantum registers $\mathsf{WX}$.
\end{itemize}
\end{definition}
\end{mdframed}

We now describe the subroutine used by the verifier to check the validity of a partially extracted quantum state against a qLTC code. This procedure ensures that the measurement results encoded in the prover's messages are consistent with the structure of the code.

\begin{mdframed}
\begin{definition}[Validity check, verifier]
\label{def:val:v}
Recall the notations introduced in Section~\ref{sec:notation}.
In addition, let $I \subseteq [M]$ such that, for every $i\in [M/m]$, either $((i - 1)\cdot m,\, i \cdot m] \subseteq I$ or $((i - 1)\cdot m,\, i \cdot m] \subseteq [M] \setminus I$ holds.

The verifier performs the following steps:
\begin{itemize}
\item The verifier samples queries $(j_1,\ldots,j_Q)\in [M]^Q$ for the code tester $\texttt{L}_W$, where $Q$ is the query complexity.
\item The verifier initializes a single-qubit register $\mathsf{T}'$ and a $Q$-qubit register $\mathsf{T}$ in state $\ket{0}$.
\item For each $k \in [Q]$, the verifier performs the following.

If $j_k \in I$, the verifier directly copies register $\mathsf{W}_{j_k}$ in register $\mathsf{T}_k$. Formally, the verifier applies $C(\mathsf{T}_k, \mathsf{W}_{j_k})$.

If $j_k \in [M] \backslash I$, the verifier ``extracts'' $\mathsf{W}_{j_k}$ from registers provided by the prover, as follows:
\begin{itemize}
\item Let $(p,q)$ be such that $j_k$ is the $q$-th element of the $p$-th block of $m$ qubits among $\mathsf{W}_1,\ldots,\mathsf{W}_M$.

\item The verifier copies register $\mathsf{X}_{p,q}$ in register $\mathsf{T}_k$. (Formally, the verifier applies $C(\mathsf{T}_k, \mathsf{X}_{p,q})$.)
\end{itemize}
\item The verifier rejects if the bits $(\mathsf{T}_1,\ldots,\mathsf{T}_Q)$ are rejected by the code tester $\texttt{L}$. Otherwise, it accepts.

Formally, the verifier applies the unitary $L_{\mathsf{T}'\mathsf{T}, r}$, where $L$ is the local tester operator as in Definition~\ref{def:tester},  measures $\mathsf{T}'$ and rejects if the outcome is $0$.
\end{itemize}
\end{definition}
\end{mdframed}

Before defining the full global measurement extraction procedure for the verifier, we recall that this process allows the verifier to sequentially extract measurement outcomes from different parts of the qLTC-encoded state, while enforcing local consistency conditions verified by the code tester $\texttt{L}_W$. We now describe the detailed steps of the verifier's actions.

\begin{mdframed}
\begin{definition}[Global measurement extraction, verifier]
\label{def:gme:v}
Recall the notations introduced in Section~\ref{sec:notation}.
 Let $ s \in \F_2^{n+1}$.
\begin{itemize}
\item If $ s_0 = 1$, the verifier first runs the local tester of the qLTC on the quantum state in register $\mathsf{W}$ and rejects if the test fails. More specifically, the verifier runs the procedure from Definition~\ref{def:val:v} with parameters $I = [M]$, $W = W$, $\mathsf{W} = \mathsf{W}$, and $\mathsf{X} = \mathsf{X}$, and rejects if this procedure rejects.
\item The verifier performs the following steps, repeating the actions described in Definition~\ref{def:lme:v} a total of $\rd=M/m$ times. In the $i$-th iteration:
\begin{itemize}
\item The verifier applies the procedure from Definition~\ref{def:lme:v} using $W=W, \mathsf{W}=\mathsf{W}_i$, and $\mathsf{X} = \mathsf{X}_i$.
\item If $ s_i = 1$, the verifier runs the procedure from Definition~\ref{def:val:v} with parameters $I = (i\cdot m,M]$, $W = W$, $\mathsf{W} = \mathsf{W}$, and $\mathsf{X} = \mathsf{X}$, and rejects if this procedure rejects.
\end{itemize}
\end{itemize}

\end{definition}
\end{mdframed}

\subsection{Main results}

\begin{lemma}[Implication of the validity test]
\label{lem:val}
Fix $W\in\{X,W\}$. Recall the notations introduced in Section~\ref{sec:notation}.
 Fix notation as in Definition~\ref{def:val:v}.
Let
 \begin{align*}
\ket{\phi}_{\mathsf{WXP}} = \sum_{p,x,w} \alpha_{p,x,w} \ket{p}_{\mathsf{P}} \ket{x}_{\mathsf{X}} (S\ket{w}_{\mathsf{W}})
\end{align*}
 be the quantum state of the system before the test defined in Definition \ref{def:val:v}.

For $x = (x_i)_{i \in [M/m]}$ and $w = (w_i)_{i \in [M/m]}$, define for each $i \in [M/m]$:
\begin{itemize}
\item If $((i - 1)\cdot m,\, i \cdot m] \subseteq I$, let $r_i = w_i$;
\item If $((i - 1)\cdot m,\, i \cdot m] \subseteq [M] \setminus I$, let $r_i = x_i$.
\end{itemize}

The rejection probability of the verifier from Definition~\ref{def:val:v} satisfies
 \begin{align*}
\Pr[\text{verifier rejects in the validity test}] \geq \kappa \cdot \sum |\alpha_{p,x,w}|^2 \cdot  \frac{\text{dist}(r, C_W)}{M} \;.
\end{align*}
\end{lemma}

\begin{proof}

Fix a query $j_k$, and let $p = \lceil j_k/m \rceil$ and $q = j_k - m(p - 1)$. If $j_k \in I$ the verifier copies the witness bit $w_{j_k}$ into $\mathsf{T}_k$; if $j_k \in [M]\setminus I$ it copies the reported bit $x_{p,q}$ into $\mathsf{T}_k$ (Definition~\ref{def:val:v}). On every query the tester therefore reads the corresponding coordinate of the decoded string $r$, so the verifier accepts exactly when the local tester accepts $r$:
\begin{align*}
\Pr[\text{verifier accepts } \ket{p,x}S\ket{w}]
= &~ \Pr[\text{local tester of
$C_W$ accepts } r] \\
\leq &~ 1 - \kappa \cdot \frac{\text{dist}(r, C_W)}{M}\;,
\end{align*}
where the last step is Lemma~\ref{lem:ltc_amp}.
Therefore,
\begin{align*}
\Pr[\text{verifier rejects}]
= &~ \sum |\alpha_{p,x,w}|^2 \cdot \Pr[\text{verifier rejects } \ket{p,x}S\ket{w}]\\
\geq &~ \kappa \cdot \sum |\alpha_{p,x,w}|^2 \cdot \frac{\text{dist}(r, C_W)}{M}\;.
\end{align*} 

\end{proof}

\begin{lemma}[Implication of the single round measurement extraction]
\label{lem:lme}
Recall the notations introduced in Section~\ref{sec:notation}.
 Fix notation as in Definition~\ref{def:gme:v}.
Let
 \begin{align*}
\ket{\phi} = \sum_{p,x,w} \alpha_{p,x,w} \ket{p}_{\mathsf{P}} \ket{x}_{\mathsf{X}} S\ket{w}_{\mathsf{W}}\;,\quad\ket{\phi'} = \sum_{p,x,w} \alpha'_{p,x,w} \ket{p}_{\mathsf{P}} \ket{x}_{\mathsf{X}} S\ket{w}_{\mathsf{W}}
\end{align*}
be the quantum state of the system respectively before and after the $i$th iteration of the single round measurement extraction in the global measurement extraction process introduced in Definition~\ref{def:gme:v}.

Let $U_i$ be the global unitary that corresponds to the verifier and the prover's joint actions during the $i$th round of the interaction.
For $x = (x_i)_{i \in [M/m]}$ and $w = (w_i)_{i \in [M/m]}$, define for each $j \in [M/m]$:
\begin{itemize}
\item If $ j \geq i $, let $r_j = w_j$;
\item If $ j < i $, let $r_j = x_j$.
\end{itemize}
Then,
 \begin{align*}
&~\| U_{i} Q_{i-1} (\ket{0}_{\mathsf{T}} \ket{\phi}) - Q_{i} U_{i} (\ket{0}_{\mathsf{T}} \ket{\phi})\|^2 \\
\leq&~ \Theta(1) \cdot \Big(\sum_{\text{dist}(r, C_W) \geq \eta_q d_{q,W} - m } |\alpha_{p,x,w}|^2 + \sum_{\text{dist}(r, C_W) \geq \eta_q d_{q,W} - m } |\alpha'_{p,x,w}|^2\Big)\; .
\end{align*}
\end{lemma}
\begin{proof}
Let
 \begin{align*}
\ket{\beta} = \sum_{ \text{dist}(r, C_W) < \eta_q d_{q,W} - m } \alpha_{p,x,w} \ket{p}_{\mathsf{P}} \ket{x}_{\mathsf{X}} (S\ket{w}_{\mathsf{W}})
\end{align*}
and
\begin{align*}
\ket{\beta'} = \sum_{ \text{dist}(r, C_W) < \eta_q d_{q,W} - m }
\alpha'_{p,x,w} \ket{p}_{\mathsf{P}} \ket{x}_{\mathsf{X}} (S\ket{w}_{\mathsf{W}})\;.
\end{align*}
 Let $\mathsf{X} = (\mathsf{X}_i)_{i\in [M/m]}$ and $\mathsf{W} = (\mathsf{W}_i)_{i\in [M/m]}$. Let $ \mathsf{W}_0 $ be a quantum register of $ m $ qubits. Let $ \mathsf{X}_0 $ be a quantum register of $ m $ qubits. Let
\begin{align*}
\mathsf{X}' &= (\mathsf{X}_1,\ldots, \mathsf{X}_{i-1}, \mathsf{X}_0,\mathsf{X}_{i+1},\ldots, \mathsf{X}_{M/m}),\\
\mathsf{W}' &= (\mathsf{W}_1,\ldots, \mathsf{W}_{i-1}, \mathsf{W}_0,\mathsf{W}_{i+1},\ldots, \mathsf{W}_{M/m}).
\end{align*}
Recall the partial composite decoding operator $D_{p,W,i}$ from Definition~\ref{def:decode:comp_p}, and let
 \begin{align*}
Q'_i=S_{\mathsf{W'}}^\dagger D_{p,W,i,\mathsf{TX'W'}}^\dagger Z_{\Tilde{f},\mathsf{T}}D_{p,W,i,\mathsf{TX'W'}}S_{\mathsf{W'}}\;,
\end{align*}
and
\begin{align*}
Q'_{i-1}=S_{\mathsf{W'}}^\dagger D_{p,W,i-1,\mathsf{TX'W'}}^\dagger Z_{\Tilde{f},\mathsf{T}} D_{p,W,i-1,\mathsf{TX'W'}}S_{\mathsf{W'}}\;.
\end{align*}
 By definition, for any quantum state $ \ket{\psi} $ on registers $\mathsf{X'}_{-i}\mathsf{W'}_{-i}$:
\begin{align*}
Q'_{i-1} (\ket{0}_{\mathsf{TX_0W_0}} \ket{\psi}) = Q'_{i} (\ket{0}_{\mathsf{TX_0W_0}} \ket{\psi})\;.
\end{align*}
This is because $D_{p,W,i-1,\mathsf{TX'W'}}$ in $Q'_{i-1}$ uses the decoding of the first $i-1$ blocks of $\mathsf{X'}$ and the last $M/m - i + 1$ blocks of $\mathsf{W'}$, while $D_{p,W,i,\mathsf{TX'W'}}$ in $Q'_i$ uses the decoding of the first $i$ blocks of $\mathsf{X'}$ and the last $M/m - i$ blocks of $\mathsf{W'}$. However, both the $i$th block of $\mathsf{X'}$ and the $i$th block of $\mathsf{W'}$ are set to $\ket{0}$, so the two operators $Q'_{i-1}$ and $Q'_i$ act identically.

Note that
 \begin{align*}
&~Q_{i-1} (\ket{0}_{\mathsf{TX_0W_0}} \ket{\beta}) \\
=&~Q'_{i-1} (\ket{0}_{\mathsf{TX_0W_0}} \ket{\beta}) .
\end{align*}
This follows since $\ket{\beta}$ is supported on $r$ such that $ \text{dist}(r, C_W) < \eta_q d_{q,W} - m $ and Lemma \ref{lem:q_noise}.
More precisely, the difference between $Q_{i-1}$ and $Q'_{i-1}$ is that the former uses the $i$th block of $\mathsf{W}$, while the latter replaces that block with $\ket{0}$. Decompose the state in the computational basis; then, applying the decoding to the first $i-1$ blocks of $\mathsf{X}$ (respectively $\mathsf{X'}$) and combining the last $M/m - i + 1$ blocks of $\mathsf{W}$ (respectively $\mathsf{W'}$) gives us a string $r$ (respectively $r'$). For an honest prover, $r$ is a codeword of $C_W$, and $r'$ is obtained from $r$ by setting the $i$th block to $0$.
As long as the two decoded strings from $r$ and $r'$ are equal,
or their respective differences from the same codeword both lie within the tolerance required by Lemma~\ref{lem:q_noise},
the actions of $Q_{i-1}$ and $Q'_{i-1}$ will be the same (since we only consider $\Tilde{f}(x)$ of the form $f(K E_W x )$). Indeed, $\mathrm{dist}(r, C_W) < \eta_q d_{q,W} - m$ implies the existence of $s \in C_W$ such that $\mathrm{dist}(r, s) < \eta_q d_{q,W} - m \leq \eta_q d_{q,W}$, and hence $\mathrm{dist}(r', s) < \eta_q d_{q,W}$.

Therefore,
 \begin{align*}
&~U_{i} Q_{i-1} (\ket{0}_{\mathsf{TX_0W_0}} \ket{\beta}) \\
=&~U_{i} Q'_{i-1} (\ket{0}_{\mathsf{TX_0W_0}} \ket{\beta}) \\
=&~ U_{i}  Q'_{i} (\ket{0}_{\mathsf{TX_0W_0}} \ket{\beta}) \\
= &~ Q'_{i} U_{i} (\ket{0}_{\mathsf{TX_0W_0}} \ket{\beta})\;,
\end{align*}

Furthermore,
 \begin{align*}
 Q'_{i} (\ket{0}_{\mathsf{TX_0W_0}} \ket{\beta'})=  Q_{i}(\ket{0}_{\mathsf{TX_0W_0}} \ket{\beta'})\;,
\end{align*}

Moreover,
 \begin{align*}
\| \ket{\beta}- \ket{\phi}\|^2 \leq \sum_{\text{dist}(r, C_W) \geq  \eta_q d_{q,W} - m } |\alpha_{p,x,w}|^2 \;,
\end{align*}
and
\begin{align*}
\| \ket{\beta'}- \ket{\phi'}\|^2 \leq \sum_{\text{dist}(r, C_W) \geq  \eta_q d_{q,W} - m } |\alpha'_{p,x,w}|^2 \;.
\end{align*}
Note that
\begin{align*}
&~\| U_{i}  \ket{\beta} - \ket{\beta'}\|^2 \\
\leq &~ \Theta(1)\cdot (
\| U_{i}  \ket{\beta} - U_{i}  \ket{\phi}\|^2 +
\| \ket{\phi'} - \ket{\beta'}\|^2
)\\
=&~ \Theta(1)\cdot (
\|  \ket{\beta} -  \ket{\phi}\|^2 +
\| \ket{\phi'} - \ket{\beta'}\|^2
)\;.
\end{align*}
Therefore,
 \begin{align*}
&~ \|U_{i} Q_{i-1} (\ket{0}_{\mathsf{TX_0W_0}} \ket{\phi}) - Q_{i} U_{i} (\ket{0}_{\mathsf{TX_0W_0}} \ket{\phi})\|^2\\
\leq &~\Theta(1)\cdot (
\|   \ket{\beta} -   \ket{\phi}\|^2 +
\| \ket{\phi'} - \ket{\beta'}\|^2
)\\
\leq &~\Theta(1)\cdot (\sum_{\text{dist}(r, C_W) \geq  \eta_q d_{q,W} - m } |\alpha_{p,x,w}|^2+\sum_{\text{dist}(r, C_W) \geq  \eta_q d_{q,W} - m } |\alpha'_{p,x,w}|^2
)\;.
\end{align*}
\end{proof}

\begin{lemma}[Implication of the global measurement extraction]
\label{lem:gme}
Use the same notation as in Lemma~\ref{lem:lme}, except that now
 \begin{align*}
\ket{\phi}_{\mathsf{WXP}} = \sum_{p,x,w} \alpha_{p,x,w} \ket{p}_{\mathsf{P}} \ket{x}_{\mathsf{X}} (S\ket{w}_{\mathsf{W}})
\end{align*}
is the quantum state of the system before the global measurement
extraction process. In particular, let $Q_i$, for $i\in\{0,\ldots,n\}$, be as defined in~\eqref{eq:def-qi}, for some $W\in\{X,Z\}$.
Suppose that for each $i$,
\begin{align*}
\Pr[\text{verifier rejects in the validity test in the $i$th iteration}] \leq \varepsilon_i\;.
\end{align*}
Suppose that during the interaction, the verifier and prover together apply the unitary $U_b$ to the entire system. Then:
 \begin{align*}
\| U_b Q_{0} (\ket{0}_{\mathsf{T}} \ket{\phi}) - Q_{n} U_b (\ket{0}_{\mathsf{T}} \ket{\phi})\|^2 \leq {\Theta\Big(\frac{\rd^2\cdot \E_i\varepsilon_i\cdot M}{\kappa\cdot (\eta_q d_{q,W} - m)}\Big)}\; ,
\end{align*}
where $ \rd = M/m$.
\end{lemma}

\begin{proof}
Let $\ket{\phi_k}_{\mathsf{WXP}}$ be the quantum state of the system after the $k$th iteration of the single round measurement extraction in the global measurement extraction process.
Let the state be written as
 \begin{align*}
\ket{\phi_k} = \sum_{p,x,w} \alpha_{k,p,x,w} \ket{p}_{\mathsf{P}} \ket{x}_{\mathsf{X}} (S\ket{w}_{\mathsf{W}})\;.
\end{align*}
For $x = (x_i)_{i \in [M/m]}$ and $w = (w_i)_{i \in [M/m]}$, define for each $j \in [M/m]$:
\begin{itemize}
\item If $ j \geq i $, let $r_j = w_j$;
\item If $ j < i $, let $r_j = x_j$.
\end{itemize}
By Lemma \ref{lem:val},
 \begin{align*}
&~\Pr[\text{verifier rejects in the validity test in the $k$th iteration}] \\
\geq&~ \kappa \cdot \sum |\alpha_{k-1,p,x,w}|^2 \cdot  \frac{\text{dist}(r, C_W)}{M}  \\
\geq&~ \kappa \cdot \sum_{\text{dist}(r, C_W) \geq \eta_q d_{q,W} - m } |\alpha_{k-1,p,x,w}|^2 \cdot  \frac{\eta_q d_{q,W}-m}{M}  .
\end{align*}
Therefore,
 \begin{align}
\sum_{\text{dist}(r, C_W) \geq \eta_q d_{q,W} - m } |\alpha_{k-1,p,x,w}|^2 \leq \Theta\Big(\frac{\varepsilon_i\cdot M}{\kappa\cdot (\eta_q d_{q,W} - m)}\Big)\;.\label{eq:lmea}
\end{align}
By Lemma \ref{lem:lme},
 \begin{align*}
&~\| U_{k} Q_{k-1} (\ket{0}_{\mathsf{T}} \ket{\phi_{k-1}}) - Q_{k} U_{k} (\ket{0}_{\mathsf{T}} \ket{\phi_{k-1}})\|^2 \\
\leq&~ \Theta(1) \cdot \Big(\sum_{\text{dist}(r, C_W) \geq \eta_q d_{q,W} - m } |\alpha_{k-1,p,x,w}|^2 + \sum_{\text{dist}(r_k, C_W) \geq \eta_q d_{q,W} - m } |\alpha_{k,p,x,w} |^2\Big)\; .
\end{align*}
Using the triangle inequality and plugging in~\eqref{eq:lmea},
 \begin{align*}
&~\| U_b Q_{0} (\ket{0}_{\mathsf{T}} \ket{\phi}) - Q_{n} U_b (\ket{0}_{\mathsf{T}} \ket{\phi})\| \\
\leq &~ \sum_{k\in[M/m]} \| U_{k} Q_{k-1} (\ket{0}_{\mathsf{T}} \ket{\phi_{k-1}}) - Q_{k} U_{k} (\ket{0}_{\mathsf{T}} \ket{\phi_{k-1}})\| \\
\leq &~  r\cdot \sqrt{\Theta\Big(\frac{\E_i \varepsilon_i\cdot M}{\kappa\cdot (\eta_q d_{q,W} - m)}\Big)} .
\end{align*}
Squaring gives the desired result:
 \begin{align*}
\| U_b Q_{0} (\ket{0}_{\mathsf{T}} \ket{\phi}) - Q_{n} U_b (\ket{0}_{\mathsf{T}} \ket{\phi})\|^2\leq {\Theta(\frac{r^2\cdot \E_i\varepsilon_i\cdot M}{\kappa\cdot (\eta_q d_{q,W} - m)})} .
\end{align*}

\end{proof}

\section{Measurement in the Computational or Hadamard Basis}
\label{sec:measurement}
To measure observables of the form $X_f$ or $Z_f$ on the physical qubits, the procedures are very similar. We continue with the $Z_f$ case as an example. After completing the Global Measurement Extraction process described in the previous section, the verifier has all the necessary information to evaluate any measurement in the $Z$ basis. To measure $Z_f$, the verifier sends the description of $f$ to the prover. The prover then uses the stored measurement outcomes from each portion and provides a PCPP proof demonstrating that the reported outcomes, when evaluated by the function $f$, yield the claimed value.

\subsection{Actions of the honest prover}

We introduce a basic quantum operator that computes a function on the string stored in a quantum register.
\begin{definition}[Function evaluation operator]
\label{def:func_o}
Let $M \in \mathbb{Z}_+$ and $f : \{0, 1\}^{M} \rightarrow \{0, 1\}$.
For $a \in \mathbb{F}_2$ and $w \in \mathbb{F}_2^{M}$, the evaluation operator $R_f$ acts as:
\begin{align*}
R_f(\ket{a} \ket{w}) = \ket{a \oplus f(w)} \ket{w},
\end{align*}
\end{definition}

We now describe how the honest prover performs measurements in either the computational basis (Z basis) or the Hadamard basis (X basis), depending on the verifier's instruction.
This procedure builds on the previously defined global measurement extraction process.
\begin{mdframed}
\begin{definition}[Measurement in the computational or Hadamard basis, honest prover]
\label{def:meas_xz:hp}
Recall the notations introduced in Section~\ref{sec:notation}.

The honest prover performs the following steps:
\begin{itemize}
\item The verifier sends the prover the description of a classical function $h : \{0, 1\}^{M} \rightarrow \{0, 1\}$.

\item The prover applies the unitary $R_{h, \mathsf{IX}'}$.
\item The prover runs the procedure in Definition \ref{def:pcpp:hp} with $C=C_W$, verification circuit $f\circ K\circ E_W$, $\mathsf{U} = \mathsf{I}$, $\mathsf{E} = \mathsf{X'}$, $\mathsf{B} = \mathsf{A}$.
\item The prover sends the verifier the quantum register $\mathsf{I}$.

\end{itemize}
\end{definition}
\end{mdframed}
\begin{remark}
The function $h$ will be specified as $f(KE_W\texttt{D}_W(\cdot))$ in a later section, but we present a general protocol here for clarity.
\end{remark}

\subsection{Actions of the verifier}

In the following, we describe the verifier's procedure for performing a global measurement in either the computational ($Z$) basis or the Hadamard ($X$) basis.
Since the measurement outcomes have already been extracted, the protocol remains the same for both bases.
However, the evaluation function $h$ differs between the two cases, even for the same function $f$.
This process can be viewed as a verification of function evaluation using a PCPP.

\begin{mdframed}
\begin{definition}[Measurement in the computational or Hadamard basis, verifier]
\label{def:meas_xz:v}
Recall the notations introduced in Section~\ref{sec:notation}
 and Definition \ref{def:meas_xz:hp}. Let $ s\in \F_2$.

The verifier performs the following steps:
\begin{itemize}
\item The verifier sends the prover the description of a classical function $h : \{0, 1\}^{M} \rightarrow \{0, 1\}$.
\item The prover sends back to the verifier the quantum registers $\mathsf{IA}$.

\item If $ s=1$, the verifier runs the procedure in Definition \ref{def:pcpp:v} with $C=C_W$, verification circuit $f\circ K\circ E_W$, $\mathsf{U} = \mathsf{I}$, $\mathsf{E} = \mathsf{X'}$, $\mathsf{B} = \mathsf{A}$. The verifier accepts if and only if this procedure accepts.
\end{itemize}
\end{definition}
\end{mdframed}

\subsection{Main results}

\begin{lemma}[Implication of the PCPP check in the measurement procedure]
\label{lem:m_zx:pcpp}
Recall the notations introduced in Section~\ref{sec:notation} and Definition \ref{def:meas_xz:hp}.
Let $\mathsf{O}=(\mathsf{P}, \mathsf{W})$ be a quantum register.
Let $\ket{\phi}_{\mathsf{OIXA}}$ be the quantum state of the system after the measurement in the computational or Hadamard basis process.
Let the state be written as
 \begin{align*}
\ket{\phi} = \sum_{o,i,x,a} \alpha_{o,i,x,a}
\ket{o}_{\mathsf{O}}
\ket{i}_{\mathsf{I}} \ket{x}_{\mathsf{X}}
\ket{a}_{\mathsf{A}}.
\end{align*}

Assume that $\delta M < \eta_W d_{q,W}$.
Then the rejection probability of the verifier is lower bounded by
 \begin{align*}
\Pr[\text{verifier rejects}] \geq  - s + \sum |\alpha_{o,i,x,a}|^2 \cdot 1\Big(h(x) \neq i\Big) .
\end{align*}
\end{lemma}
\begin{proof}

 The verifier's check (Definition~\ref{def:meas_xz:v}) runs the procedure of Definition~\ref{def:pcpp:v} with $C_P=C_W$, $\texttt{D}_P=\texttt{D}_W$, $\texttt{T}_P=\texttt{T}_W$, verification circuit $f\circ K\circ E_W$, and registers $\mathsf{U}=\mathsf{I}$, $\mathsf{E}=\mathsf{X}$, $\mathsf{B}=\mathsf{A}$; the contents of $\mathsf{U}$ and $\mathsf{E}$ are then $i$ and $x$. Composing the verification circuit with the decoder gives $(f\circ K\circ E_W)(\texttt{D}_W(x))=f(KE_W\texttt{D}_W(x))=h(x)$.

It remains to verify the condition of Lemma~\ref{lem:pcpp}, which under this substitution states that $h(x)=b$ for every $b\in\F_2$ and every $x$ with $\dist(x,S_b)\le\delta$, where
\begin{align*}
S_b=\{\,x \mid 1\big(h(x)=b\big)\wedge \texttt{T}_W(x)=1\,\}.
\end{align*}
Fix $b$ and any $x$ with $\dist(x,S_b)\le\delta$, and choose $x^*\in S_b$ with $\dist(x,x^*)\le\delta$.
\begin{itemize}
\item Since $x^*\in S_b$, we have $\texttt{T}_W(x^*)=1$, so $x^*$ is a codeword of $C_W$ and $x^*=E_W v^*$ for some $v^*$.
\item Since $\dist(x,x^*)\le\delta$ and $x,x^*\in\F_2^M$, the Hamming weight satisfies $|x-x^*|=\dist(x,x^*)\cdot M\le\delta M<\eta_W d_{q,W}$, by the assumption $\delta M<\eta_W d_{q,W}$.
\item Write $x=x^*+p$ with $|p|<\eta_W d_{q,W}$. By Lemma~\ref{lem:q_noise}, $KE_W\texttt{D}_W(x)=KE_W\texttt{D}_W(p+E_W v^*)=KE_W v^*=Kx^*$; the same lemma with $p=0$ gives $KE_W\texttt{D}_W(x^*)=Kx^*$.
\item Hence $h(x)=f(KE_W\texttt{D}_W(x))=f(Kx^*)=f(KE_W\texttt{D}_W(x^*))=h(x^*)=b$, where the last equality holds since $x^*\in S_b$.
\end{itemize}
Thus $h(x)=b$ whenever $\dist(x,S_b)\le\delta$, which verifies the condition. The claim follows from Lemma~\ref{lem:pcpp} with the substitutions above:
\begin{align*}
\Pr[\text{verifier rejects}] \geq  - s + \sum |\alpha_{o,i,x,a}|^2 \cdot 1\Big( h(x) \neq i\Big) .
\end{align*} 
\end{proof}

\begin{lemma}[Implication of the measurement in the computational or Hadamard basis]
\label{lem:measxz}
Recall the notations introduced in Section~\ref{sec:notation} and Definition \ref{def:meas_xz:hp}.
Let $\mathsf{O}=(\mathsf{P}, \mathsf{W})$ be a quantum register.
Let $\ket{\phi}_{\mathsf{OIXA}}$ be the quantum state of the system after the measurement in the computational or Hadamard basis process.
Let the state be written as
 \begin{align*}
\ket{\phi} = \sum_{o,i,x,a} \alpha_{o,i,x,a}
\ket{o}_{\mathsf{O}}
\ket{i}_{\mathsf{I}} \ket{x}_{\mathsf{X}}
\ket{a}_{\mathsf{A}}.
\end{align*}

If $h(x) = f(KE_W\texttt{D}_W(x))$, then:
 \begin{align*}
\| Z_{\mathsf{I}}\ket{\phi} -
Q_\rd \ket{\phi}
\|^2 \leq  \Theta(1) \cdot
\sum |\alpha_{o,i,x,a}|^2 \cdot 1\Big(h(x)\neq i\Big)
.
\end{align*}
\end{lemma}
\begin{proof}

Let
 \begin{align*}
\ket{\beta} = \sum_{o,i,x,a} 1\Big(h(x) = i\Big)
 \cdot \alpha_{o,i,x,a}
\ket{o}_{\mathsf{O}}
\ket{i}_{\mathsf{I}} \ket{x}_{\mathsf{X}}
\ket{a}_{\mathsf{A}}.
\end{align*}

Note that
\begin{align*}
Z_{\mathsf{I}}\ket{\beta} =
Q_\rd \ket{\beta}.
\end{align*}

Therefore,
 \begin{align*}
&~\| Z_{\mathsf{I}}\ket{\phi} -
Q_\rd \ket{\phi}
\|^2 \\
\leq &~ \Theta(1) \cdot (\| \ket{\phi} -
 \ket{\beta }
\|^2 + \|Z_{\mathsf{I}}\ket{\beta} -
Q_\rd \ket{\beta}\|^2)\\
\leq &~ \Theta(1) \cdot
\sum |\alpha_{o,i,x,a}|^2 \cdot 1\Big(h(x) \neq i\Big)   .
\end{align*}

\end{proof}

\section{QIPCP protocol for QMA}
\label{sec:energy}

In this section, we describe our protocol for estimating the energy of a dual-basis projector Hamiltonian, and thereby distinguishing between the yes and no instances of the corresponding QMA-complete problem. This will establish our main result, restated as Theorem~\ref{thm:sqiop} below.

To estimate the energy, we sample a term according to its coefficient in the Hamiltonian (each term being a projector) and then project the witness state according to that projector by implementing the corresponding measurement.
The measurement procedures required for this have already been developed in Sections~\ref{sec:gme} and~\ref{sec:measurement}. Here we combine the results from these two sections and apply them to the specific case where the function $f$ arises from a projector in the support of a dual-basis projector Hamiltonian (Definition~\ref{def:d_ham}).

\subsection{Setup}
\label{sec:setup}

We summarize the main notation, parameters and assumptions that are made throughout this section.

Firstly, $H=\E_{i} H_i$ s a dual-basis projector Hamiltonian, as introduced in Definition~\ref{def:d_ham}. According to Theorem~\ref{thm:dual_ham}, it is $\QMA$-hard to distinguish between the case where $H$ has smallest eigenvalue exponentially close to $0$, or bounded away from $0$ by a constant. We let $B_q$ be the number of qubits on which $H$ acts.

Secondly, a quantum LTC code $C$ is chosen, which encodes\ $B_q$ qubits into $M$ qubits, has distance $d_q$ and locality $q$. We will use that  $C_q$ admits an efficient decoder $\tD_q$ for errors of weight up to $\eta_q d_q$, where $d_q$ is the quantum distance of the code. These parameters satisfy the conditions of Theorem~\ref{thm:qltc} and Fact~\ref{fact:decoder}.

The $M$ qubits are divided into $r$ groups of $m$ qubits, where $r$ is chosen sufficiently large that the condition~\eqref{ass:etad} on $m$ holds; this can be satisfied with $r=\poly\log(B_q)$.

Further notation and parameters associated with the codes is summarized in Section~\ref{sec:notation-codes}. The quantum LTC code is used to encode the witness.

Thirdly, we fix a PCPP verifier and corresponding proof generating algorithm for the \textsc{CktVal} problem. The PCPP has proximity parameter $\delta$ and soundness error $s$. (Further notation associated with the PCPP is recalled in Section~\ref{sec:notation-pcpp}.)

The main assumption relating these parameters are the following.
\begin{itemize}
\item

 The PCPP proximity parameter $\delta$ is chosen in such a way that
\begin{equation}\label{ass:cond}
\delta M < \eta_W d_{q,W}\;.
\end{equation}
Intuitively, this relation is needed for the final PCPP verifier to interact well with the quantum code's classical part $C_W$; it first arises from the use of Lemma~\ref{lem:pcpp}. A report accepted by the PCPP lies within relative distance $\delta$ of a codeword of $C_W$, hence differs from it in at most $\delta M$ coordinates; when $\delta M<\eta_W d_{q,W}$ the decoder $\texttt{D}_W$ recovers its logical value (Fact~\ref{fact:decoder}). Since $\eta_W d_{q,W}=\Theta(M/\poly\log B_q)$ (Theorem~\ref{thm:qltc}, Fact~\ref{fact:decoder}), dividing by $M$ shows that $\delta M<\eta_W d_{q,W}$ is equivalent to $\delta<\eta_W d_{q,W}/M=\Theta(1/\poly\log B_q)$, which is satisfied by choosing a sufficiently small $\delta$ with $\delta^{-1}=\Theta(\poly\log B_q)$. 
\item The number $m$ of qubits of the qLTC-encoded witness exchanged in each round is chosen such that
    \begin{equation}\label{ass:etad}
    \eta_q d_q \geq 2m\;.
    \end{equation}
    This can be achieved by choosing $r$ to be sufficiently large, $r=\polylog(B_q)$ will suffice given the quantum LTC code parameters (which are all a multiplicative $\polylog$ away from optimal). Intuitively this condition ensures that the prover manipulates too few qubits at a time to perform an undetected attack (in the basis currently being verified).
    \item The soundness error $s$ of the PCPP is a small enough constant. This is because this soundness error appears as an additive term in the final soundness error.
\end{itemize}

\subsection{Actions of the honest prover}

We define the behavior of an honest prover during the energy test, which consists of two phases: the global measurement extraction phase and the measurement phase.

\begin{mdframed}
\begin{definition}[Energy test, honest prover]
\label{def:e:hp}

The honest prover performs the following steps:
\begin{itemize}
\item The verifier
sends the basis $W\in\{X, Z\}$ to the prover.
\item The honest prover performs the actions described in Definition~\ref{def:gme:hp}.
\item The honest prover performs the actions described in Definition~\ref{def:meas_xz:hp}.
\end{itemize}
\end{definition}
\end{mdframed}

\subsection{Actions of the verifier}
We define the verifier's strategy in the energy test, which is designed to verify that the witness state has low energy.
In what follows, the case $s_c = 3$ corresponds to the energy test, whose success probability is proportional to the energy of the Hamiltonian.
The case $s_c = 1$ is intended to enforce honest behavior from the prover during the global measurement extraction phase, while the case $s_c = 2$ ensures honest behavior during the measurement phase.

\begin{mdframed}
\begin{definition}[Energy test, verifier]
\label{def:e:verifier}
Recall the notations introduced in Section~\ref{sec:notation}.
 Let $H=\E_{i} H_i$ be as  in Definition \ref{def:d_ham}, where $H_i =  {I + W_{i, f_i}}/{2}$, $W_i$ is either $X$ or $Z$, $f_i:\F_2^{B_q} \rightarrow \F_2$ is a boolean function.

The verifier performs the following steps:
\begin{itemize}
\item The verifier samples the phase of the protocol to test by choosing $s_c \in_R [3]$,
and selects a round to test in the global measurement extraction phase by choosing $s'_g \in_R [n]$ with probability $1/2$, and setting $s'_g = 0$ with probability $1/2$.
If the verifier tests the global measurement extraction phase ($s_c = 1$), it sets $s_g = e_{s'_g}$, where $e_{s'_g} \in \{0,1\}^{n+1}$ is the bit string with all entries $0$ except the $s'_g$-th bit, which is $1$.
If the verifier does not test during the global measurement extraction phase ($s_c \ne 1$), it sets $s_g = 0^{n+1}$.
\item The verifier samples a term $H_i$ by selecting the index $i$ uniformly at random. The verifier
sends the basis $W_i$ to the prover.
\item The verifier performs the actions described in Definition~\ref{def:gme:v} with $ W =W_i$, $ s = s_g$.
\item If $s_c = 1$, the verifier accepts if and only if the preceding procedure accepts.
\item If $ s_c = 2$, the verifier performs the actions described in Definition~\ref{def:meas_xz:v} with $ h(x) = f_i(KE_W\texttt{D}_W(x) )$ and $s=1$. The verifier accepts if and only if the procedure accept.
\item If $ s_c =3$, the verifier performs the actions described in Definition~\ref{def:meas_xz:v} with $ h(x) = f_i(KE_W\texttt{D}_W(x) )$ and  $s=0$. The verifier measures the register $\mathsf{I}$ and accept if and only if the measurement result is $ 0 $.
\end{itemize}
\end{definition}
\end{mdframed}

\subsection{Main results}

\begin{theorem}\label{thm:sqiop}
$\QMA\subseteq \QIPCP(O(\poly), O(\poly\log),O(\poly\log))$   .
\end{theorem}

\begin{proof}

Let $\texttt{V}$ be the verifier in Definition~\ref{def:e:verifier}. Assume that the codes used for the protocol have been chosen in such a way that $\delta M < \eta_W d_{q,W}$, where $\delta$ is the proximity error of the PCPP and $\eta_W d_{q,W}$ the decoding radius of the quantum code's classical part $C_W$, as in condition~\eqref{ass:cond}. This can be achieved by choosing a good enough PCPP and quantum code.

We will show that $\texttt{V}$ is a $\poly\log$-message
quantum interactive probabilistically checkable proof system for $\mathcal{DLH}(\poly, 2^{-p(n)}, c_\LH)$ (see Definition~\ref{def:dlh}), where $c_\LH > 0$ is the constant from Corollary~\ref{cor:dual_ham_amp}. It follows that
\[
\mathcal{DLH}(\poly, 2^{-p(n)}, c_\LH) \in \QIPCP(O(\poly), O(\poly\log), O(\poly\log))\;.
\]
By Corollary~\ref{cor:dual_ham_amp}, $\mathcal{DLH}(\poly, 2^{-p(n)}, c_\LH)$ is $\QMA$-complete. Therefore,
\[
\QMA \subseteq \QIPCP(O(\poly), O(\poly\log), O(\poly\log))\;.
\]
We now show the required properties of $\texttt{V}$. Let $n$ be the instance size, i.e., the description length of a $\mathcal{DLH}$ instance $H$. In particular, the number of qubits $B_q$ on which $H$ acts satisfies $B_q\leq n$. For the construction of the QIPCP, we use parameters as introduced in Section~\ref{sec:setup}.

First, by inspecting Definitions~\ref{def:e:verifier} and~\ref{def:qiop}, it is clear that $\texttt{V}$ has polynomial-time complexity.

\medskip
{\bf Rounds:}
By Definitions~\ref{def:e:verifier} and~\ref{def:qiop}, the verifier and the prover exchange a total of $2M/m + 2$ messages.

Recall that $M$ is the length of the qLTC codeword, and $m$ is chosen to be smaller than $\eta_q d_q$ by at least a constant factor.
By Theorem~\ref{thm:qltc}, $d_q = M / \poly\log(M)$, and by Fact~\ref{fact:decoder}, $\eta_q = 1 / \poly\log(M)$.
Moreover, the qLTC code has a linear code length $M=O(n)$.
Therefore,
\[
2M/m + 2 = \poly\log(n)\;.
\]

\medskip
{\bf Communication complexity:}
We compute the communication complexity according to Definitions~\ref{def:qiop} and~\ref{def:e:verifier}.
In each of the first $M/m$ rounds, the verifier sends a quantum message of length $m$ (plus $1$ classical bit for the basis choice).
 The prover's message consists of the $m$ measurement outcomes of the block (the reported block).

The communication in the last round is $l_w(\poly(M))$.
Here, the outer $l_w$ comes from the PCPP proof, while $M$ corresponds to the size of the implicit input (i.e., the input to the circuit being verified, the length-$M$ report).
This circuit decodes the report with $\texttt{D}_W$, maps it to its logical form via $K E_W$, and evaluates the logical Boolean function corresponding to the sampled projector.
Since decoding with $\texttt{D}_W$, transforming the Boolean function into its logical form, and evaluating the Boolean function itself are all polynomial-time operations, the PCPP also has an explicit input (the circuit being verified) of size $\poly(M)$.
Therefore, the total communication complexity is $\poly(n)$.

\medskip
{\bf Query complexity:}
By Definitions~\ref{def:qiop} and~\ref{def:e:verifier}, the query complexity in each of the first $M/m$ rounds arises from the local tester for the $W$-code (in the qLTC).

We analyze the query complexity of a single invocation of the local tester for the $W$-code.
Let $\varepsilon=\Theta(1)$ be a small constant that will be specified later. Note that the soundness parameter of the LTC is $\rho=\Theta(1/\poly\log(n))$. Define
 \begin{equation}\label{eq:energy-6}
\kappa = \Theta\left(\frac{\rd^2 M}{\eta_q d_{q,W} - m}\right) = \poly\log(n)\;,
\end{equation}
since $\eta_q d_{q,W} \geq \eta_q d_{q} \geq 2m$ by assumption~\eqref{ass:etad} and $M/m=\rd=\poly\log(n)$.
By setting $R = \Theta(1)$ in Lemma~\ref{lem:ltc_amp}, we obtain the query complexity of the local tester for the $W$ codes of the qLTC as
 \[{\color{revisionfinal}
Q = \Theta(1) \cdot \frac{\kappa}{\rho} \cdot \frac{\log(1/\varepsilon)}{1 - \varepsilon} = \poly\log(n)\;.
}\]
 Therefore, in each of the first $\rd = M/m$ rounds, the total query complexity is that of the local tester:
\[{\color{revisionfinal}
\poly\log(n)\;.
}\]
The query complexity in the final round is due to the PCPP verifier, which has soundness error $s=\Theta(\eps)$, and is thus at most
\[
O\left(\frac{\log(s)}{\log(1 - \Omega(\delta))}\right) = \poly\log(n)\;,
\]
which follows from the application of Lemma~\ref{lem:pcpp_amp}.
Thus, the overall query complexity of the verifier is $\poly\log(n)$.

\medskip
{\bf Completeness:}
By Definition~\ref{def:e:hp}, note that the honest prover uses a qLTC encoding with no error, so the local tester always accepts. Moreover, the PCPP statement in the measurement phase verifies a valid computation; hence, the corresponding instance is always a yes-instance. By the completeness of the PCPP verifier, the PCPP verifier in this phase succeeds with probability~$1$. Therefore, all tests except the energy test pass with probability~$1$.

Let $\ket{\psi}_{\mathsf{OIXA}}$ denote the quantum state of the system after the measurement in the computational or Hadamard basis, as defined in Lemma~\ref{lem:measxz}, and let $\ket{\phi}_{\mathsf{WXP}}$ be the quantum state of the system before the global measurement extraction, as defined in Lemma~\ref{lem:gme}.

If $s_c = 3$, then by Definition~\ref{def:e:verifier}, the energy test passes if and only if the measurement result of the register $\mathsf{I}$ is $0$ with the corresponding probability:
\begin{align*}
 \E_i \big\| \frac{I + Z_\mathsf{I}}{2} \ket{\psi} \big\|^2
=&~ \E_i \big\| \frac{I + Q_r}{2} \ket{\psi} \big\|^2 \\
=&~ \E_i \big\| \frac{I + Q_0}{2} \ket{\phi} \big\|^2 \\
=&~ \E_i \big\| (I - W_{i, f_i}) \ket{\phi} \big\|^2 \\
=&~ \E_i \big\| (I - H_i) \ket{\phi} \big\|^2 \\
=&~ \bra{\phi} (I - H) \ket{\phi}\;,
\end{align*}
where the first step follows from Lemma~\ref{lem:m_zx:pcpp}, Lemma~\ref{lem:measxz}, and the fact that the verifier never rejects in the measurement phase; the second step follows from Lemma~\ref{lem:gme} and the verifier never rejecting in the global measurement extraction phase; and the third step follows from Lemma~\ref{lem:DO2Dq} and the fact that the quantum witness is a qLTC encoding.
As a result,
\begin{align*}
\Pr[\text{\texttt{V} accepts}]
=&~ \Pr[\text{\texttt{V} accepts}~|~s_c = 1]\cdot \Pr[s_c = 1] + \Pr[\text{\texttt{V} accepts}~|~s_c = 2]\cdot \Pr[s_c = 2] \\
&~\qquad+\Pr[\text{\texttt{V} accepts}~|~s_c = 3]\cdot \Pr[s_c = 3]\\
=&~ \frac{1}{3} + \frac{1}{3} + \frac{1}{3} \bra{\phi} (I - H) \ket{\phi}\\
\geq&~ 1 - \Theta(2^{-p(n)})\;.
\end{align*}

\medskip
{\bf Soundness:}
Let $\ket{\psi}_{\mathsf{OIXA}}$ denote the quantum state of the system after the measurement in the computational or Hadamard basis, as defined in Lemma~\ref{lem:measxz}. This corresponds to the state of the quantum system after the verifier performs the procedure described in Definition~\ref{def:meas_xz:v} within Definition~\ref{def:e:verifier}, in both cases $s_c = 2$ and $s_c = 3$.
Let $\ket{\phi}_{\mathsf{WXP}}$ denote the quantum state of the system before the global measurement extraction procedure, as defined in Lemma~\ref{lem:gme}, i.e., the state of the quantum system before the verifier executes the procedure from Definition~\ref{def:gme:v} in Definition~\ref{def:e:verifier}.
Then,
\begin{align*}
\Pr[\text{energy test accepts}~|~s_c=3] =&~ \E_i \Big\|\frac{I+Z_{\mathsf{I}}}{2}\ket{\psi}\Big\|^2\;.
\end{align*}
Note that for any vectors $a,b$ and $k>0$,
\begin{align*}
\|a+b\|^2 \leq \frac{1+k^2}{k^2} \|a\|^2 + (1+k^2) \|b\|^2\;.
\end{align*}
Using this with $k=10$,
\begin{align*}
\Big\|\frac{I+Z_{\mathsf{I}}}{2}\ket{\psi}\Big\|^2
\leq&~ 1.01\Big\|\frac{I+Q_r}{2} \ket{\psi} \Big\|^2 + 101 \Big\| \frac{Z_{\mathsf{I}}\ket{\psi} - Q_r \ket{\psi}}{2} \Big\|^2\;.
\end{align*}
Using the same reasoning,
\begin{align*}
\Big\|\frac{I+Q_r}{2} \ket{\psi}\Big\|^2
\leq&~ 1.01\Big\|U_b \frac{I+Q_0}{2} (\ket{0}_{\mathsf{T}} \ket{\phi})\Big\|^2 + 101 \Big\|\frac{(U_b Q_{0} - Q_{r} U_b )(\ket{0}_{\mathsf{T}} \ket{\phi})}{2}\Big\|^2\;.
\end{align*}
Finally,
\begin{align*}
\Big\|\frac{I+Q_{0} }{2}(\ket{0}_{\mathsf{T}} \ket{\phi})\Big\|^2
\leq&~ 1.01\|(I-W_{f_i})(\ket{0}_{\mathsf{T}} \ket{\phi})\|^2 + 101 \Big\|\frac{(Q_{0}  - (I-2W_{f_i}))(\ket{0}_{\mathsf{T}} \ket{\phi})}{2}\Big\|^2\;.
\end{align*}
Summarizing the preceding equations,
\begin{align}\label{eq:energy-1}
\Pr[&\text{energy test accepts}~|~s_c=3]\notag\\
&\leq 101 \E_i\Big\| \frac{Z_{\mathsf{I}}\ket{\psi} - Q_r \ket{\psi}}{2} \Big\|^2 + 1.01\cdot 101 \Big\|\frac{(U_b Q_{0} - Q_{r} U_b )(\ket{0}_{\mathsf{T}} \ket{\phi})}{2}\Big\|^2  \notag\\
&\quad+    1.01^2\Big(  1.01\|(I-W_{f_i})(\ket{0}_{\mathsf{T}} \ket{\phi})\|^2 + 101 \Big\|\frac{(Q_{0}  - (I-2W_{f_i}))(\ket{0}_{\mathsf{T}} \ket{\phi})}{2}\Big\|^2\;\Big) \;.
\end{align}
Assume that
\begin{align}
\Pr[\text{global measurement extraction accepts}~|~ s_c=1, s'_g = 0] \geq&~ 1-\eps\;,\notag\\
\Pr[\text{global measurement extraction accepts}~|~ s_c=1, s'_g \neq 0] \geq&~ 1-\eps\;,\notag\\
\Pr[\text{measurement phase accepts}~|~ s_c=2] \geq&~ 1-\eps\;,\label{eq:energy-9}
\end{align}
where $\eps$ is the previously introduced small constant that will be specified later. By Lemma~\ref{lem:m_zx:pcpp} and Lemma~\ref{lem:measxz},
\begin{equation}\label{eq:energy-2}
\|Z_{\mathsf{I}}\ket{\psi} - Q_r \ket{\psi}\|^2 \leq \Theta(s+\eps)\;.
\end{equation}
Moreover, suppose $\Pr[\text{global measurement extraction accepts}~|~ s_c=1, s'_g = i] = 1-\eps_i$ for some $\eps_i$, then by Lemma~\ref{lem:gme},
 \begin{equation}\label{eq:energy-3}
\| U_b Q_{0} (\ket{0}_{\mathsf{T}} \ket{\phi}) - Q_{r} U_b (\ket{0}_{\mathsf{T}} \ket{\phi})\|^2
\leq {\Theta\Big(\frac{\rd^2\cdot \E_i\varepsilon_i\cdot M}{\kappa\cdot (\eta_q d_{q,W} - m)}\Big)}\;.
\end{equation}
Let
\begin{align*}
\ket{\phi} = \sum_{(a,b)\in {\cal S}} \alpha_{a,b} X(a)Z(b) \ket{\eta_{a,b}}\;.
\end{align*}
By Lemma~\ref{lem:DO2Dq},
\begin{align*}
\|Q_{0} (\ket{0}_{\mathsf{T}} \ket{\phi}) - (I-2W_{f_i})(\ket{0}_{\mathsf{T}} \ket{\phi})\|^2
\leq \Theta(1)\cdot \sum_{(a,b)\in {\cal S}} 1(wt_{\text{cent}}(X(a)Z(b)) \geq \eta_q d_q)\cdot |\alpha_{a,b}|^2\;.
\end{align*}
Let $c$ be the number of repetitions of the local testing verifier from Definition~\ref{def:qltc-local-testing}, which will be chosen below in~\eqref{eq:def-ca}. By Lemma~\ref{lem:qltclocaltest},
\begin{align*}
1-\eps \leq&~ \sum_{(a,b)\in {\cal S}} |\alpha_{a,b}|^2\cdot \Big(1  - \frac{ \rho }{2M}\cdot wt_{\text{cent}}(X(a)Z(b)) \Big)^c \\
\leq&~ \sum_{(a,b)\in {\cal S}} 1(wt_{\text{cent}}(X(a)Z(b)) \geq \eta_q d_q)\cdot |\alpha_{a,b}|^2\cdot \Big(1  - \frac{ \rho }{2M}\cdot \eta_q d_q\Big)^c \\
&~ + \sum_{(a,b)\in {\cal S}} 1(wt_{\text{cent}}(X(a)Z(b)) < \eta_q d_q)\cdot |\alpha_{a,b}|^2\;.
\end{align*}
Therefore,
\begin{equation}\label{eq:energy-4}
\|Q_{0} (\ket{0}_{\mathsf{T}} \ket{\phi}) - (I-2W_{f_i})(\ket{0}_{\mathsf{T}} \ket{\phi})\|^2
\leq \Theta\Big( \eps\Big(1  - \frac{ \rho }{2M}\cdot \eta_q d_q\Big)^c \Big)^{-1}\Big)\;.
\end{equation}
Plugging the bounds~\eqref{eq:energy-2},~\eqref{eq:energy-3} and~\eqref{eq:energy-4} in~\eqref{eq:energy-1} and using the assumptions~\eqref{eq:energy-9} we have obtained that
 \begin{align}
\Pr[\text{energy test accepts}]
\leq&~ 1.1 \cdot \bra{\phi} (I-H) \ket{\phi}
+ \Theta(s+\eps)
+ \Theta\Big(\frac{\rd^2\cdot \E_i\varepsilon_i\cdot M}{\kappa\cdot (\eta_q d_{q,W} - m)}\Big) \notag\\
&~\qquad + \Theta\Big( \eps\Big(1 - \Big(1  - \frac{ \rho }{2M}\cdot \eta_q d_q\Big)^c\Big)^{-1}\Big)\;.\label{eq:energy-8}
\end{align}
Note that
\begin{align*}
\Pr[\text{global measurement extraction accepts}~|~ s_c=1, s'_g \neq 0] = 1 - \E_i\varepsilon_i\;.
\end{align*}
Choose
\begin{align}
c :=&~ \frac{\Theta(1)}{-\log(1  - \frac{ \rho }{2M}\cdot \eta_q d_q)} = \poly\log(n)\;.\label{eq:def-ca}
\end{align}
Plugging back into~\eqref{eq:energy-8} and using also the choice of $\kappa$ in~\eqref{eq:energy-6},
\begin{align*}
\Pr[\text{energy test accepts}]
\leq&~ 1.1 \cdot \bra{\phi} (I-H) \ket{\phi} + \Theta(\eps)\\
\leq&~ 1.1  - 1.1 c_\LH + \Theta(\eps)\;.
\end{align*}
If any of the following fails:
\begin{align*}
\Pr[\text{global measurement extraction accepts}~|~ s_c=1, s'_g = 0] \geq&~ 1-\eps\;, \\
\Pr[\text{global measurement extraction accepts}~|~ s_c=1, s'_g \neq 0] \geq&~ 1-\eps\;, \\
\Pr[\text{measurement phase accepts}~|~ s_c=2] \geq&~ 1-\eps\;,
\end{align*}
then
\begin{align*}
\Pr[\text{\texttt{V} accepts}] \leq 1 - \Theta(\eps)\;.
\end{align*}
Note that we can take $c_\LH$ in Corollary~\ref{cor:dual_ham_amp} to be sufficiently large, and in particular strictly larger than $0.1$.
As a result, in all cases, we get that $\Pr[\text{\texttt{V} accepts}]$ is smaller than some constant less than $1$, as desired.

By applying sequential repetition a constant number of times (since $c_\LH = \Theta(1),\eps = \Theta(1)$), we obtain a protocol in which the verifier accepts yes-instances with probability greater than $2/3$ and accepts no-instances with probability less than $1/3$.

\end{proof}

\newpage
\onecolumn
\appendix

\bibliographystyle{alpha}
\newcommand{\etalchar}[1]{$^{#1}$}

\end{document}